\documentclass[a4paper,UKenglish,cleveref, autoref]{lipics-v2019}

\usepackage{amssymb}
\nolinenumbers
\title{Relational Width of  First-Order Expansions of Homogeneous Graphs  with Bounded Strict Width} 

\titlerunning{Relational Width of  FO-Expansions of Homogeneous Graphs with BSW}

\author{Micha{\l} Wrona}{Theoretical Computer Science Department,
	Faculty of Mathematics and Computer Science, Jagiellonian University, Poland}{wrona@tcs.uj.edu.pl}{https://orcid.org/0000-0002-2723-0768}{}

\authorrunning{M. Wrona}

\Copyright{Micha{\l} Wrona}



\category{}

\relatedversion{}

\supplement{}

\funding{The research was partially supported by NCN grant number 2014/14/A/ST6/00138.}

\EventEditors{Christophe Paul and Markus Bl\"{a}ser}
\EventNoEds{2}
\EventLongTitle{37th International Symposium on Theoretical Aspects of Computer Science (STACS 2020)}
\EventShortTitle{STACS 2020}
\EventAcronym{STACS}
\EventYear{2020}
\EventDate{March 10--13, 2020}
\EventLocation{Montpellier, France}
\EventLogo{}
\SeriesVolume{154}
\ArticleNo{35}

\theoremstyle{plain}
\newtheorem{question}{Question}

\theoremstyle{plain}

\theoremstyle{plain}
\newtheorem{observation}[theorem]{Observation}

\theoremstyle{plain}
\newtheorem{fact}[theorem]{Fact}

\DeclareMathOperator{\csp}{CSP}
\DeclareMathOperator{\Aut}{Aut}
\DeclareMathOperator{\maxbound}{\mathbb{L}}

\DeclareMathOperator{\instance}{\mathcal{I}}

\DeclareMathOperator{\structA}{\mathbb{A}}
\DeclareMathOperator{\structB}{\mathbb{B}}
\DeclareMathOperator{\solution}{\mathbf{s}}

\DeclareMathOperator{\V}{\mathcal{V}}

\DeclareMathOperator{\Q}{\mathbb{Q}}
\DeclareMathOperator{\N}{\mathbb{N}}
\DeclareMathOperator{\bounds}{\mathcal{F}}
\DeclareMathOperator{\homograph}{\mathcal{H}}

\DeclareMathOperator{\OP}{OP}
\DeclareMathOperator{\PO}{PO}
\DeclareMathOperator{\EN}{EN}
\DeclareMathOperator{\NE}{NE}
\DeclareMathOperator{\EE}{EE}
\DeclareMathOperator{\NN}{NN}
\DeclareMathOperator{\EEq}{E\!\!=}
\DeclareMathOperator{\NEq}{N\!\!=}
\DeclareMathOperator{\EqE}{=\!\!E}
\DeclareMathOperator{\EqN}{=\!\!N}
\DeclareMathOperator{\NeqEq}{\neq =}
\DeclareMathOperator{\EqNeq}{= \neq}
\DeclareMathOperator{\NeqNeq}{\neq \neq}

\DeclareMathOperator{\uuE}{\underline{\underline{E}}}
\DeclareMathOperator{\uuN}{\underline{\underline{N}}}
\DeclareMathOperator{\uuO}{\underline{\underline{O}}}
\DeclareMathOperator{\uuOone}{\underline{\underline{O_1}}}
\DeclareMathOperator{\uuOtwo}{\underline{\underline{O_2}}}

\begin{document}

\maketitle

\begin{abstract}
Solving the algebraic dichotomy conjecture for constraint satisfaction problems over structures first-order definable in countably infinite finitely bounded homogeneous structures requires understanding the applicability of local-consistency methods in this setting. 
We study the amount of consistency (measured by relational width) needed  to solve $\csp(\mathbb{A})$
for first-order expansions  $\mathbb{A}$ of countably infinite homogeneous graphs $\homograph := (A; E)$, which happen all to be finitely bounded. We study our problem for structures
$\mathbb{A}$ that additionally have bounded strict width, i.e., for which establishing local consistency of an instance of $\csp(\mathbb{A})$ not only decides 
if there is a solution but also ensures that every solution may be  obtained from a locally consistent instance by greedily assigning values to variables, without backtracking.

Our main result is that the structures $\mathbb{A}$ under consideration have  relational width exactly $(2, \maxbound_{\homograph})$ where $\maxbound_{\homograph}$ is the maximal size of a forbidden subgraph of $\homograph$,
but not smaller than $3$. It beats the upper bound: $(2 m, 3 m)$ where $m = \max(\textrm{arity}(\mathbb{A})+1, \maxbound, 3)$ and
$\textrm{arity}(\mathbb{A})$ is the largest arity of a relation in $\mathbb{A}$,
which follows from a sufficient condition implying bounded relational width given in~\cite{UnaryDichotomy}.
Since $\maxbound_{\homograph}$ may be arbitrarily large, our result contrasts the collapse of the relational bounded width hierarchy for finite structures $\mathbb{A}$, whose relational width, if finite,
is always at most $(2,3)$. 
\end{abstract}

\section{Introduction}

The constraint satisfaction problem (CSP) is one of the most important problems in theoretical and applied computer science and at the same time it is a general framework  in which many other computational problems may be 
formalized. Given a number of constraints imposed on variables one asks if there is a global solution, i.e., a function
assigning values to variables so that all the constraints are simultaneously satisfied.
Boolean satisfiability and graph colouring are among the most prominent examples of NP-hard problems that can be formalized as CSPs and hence the CSP is NP-hard in general. Thus, one considers  
the  problem $\csp(\structA)$ parametrized by a relational structure (called also a constraint language, a language or a template) $\structA$. (In this paper, $\structA$ is always over a finite signature).
A longstanding open problem in this area was to verify the Feder-Vardi~\cite{FederVardi} conjecture
which states that for every finite $\structA$ the problem 
$\csp(\structA)$ is either in P or it is NP-complete.  After over thirty years of work and a number of
important partial results this so-called Dichotomy Conjecture was confirmed independently in~\cite{ZhukDichotomy} and~\cite{BulatovDichotomy}. In both cases the proof was carried out in the so-called  universal-algebraic approach to the complexity of CSPs~\cite{JeavonsClosure,JBK}. The approach not only provided appropriate tools but also suggested the delineation. This so-called algebraic dichotomy conjecture~\cite{JBK} saying that $\csp(\structA)$ is hard under the condition that the algebra corresponding to $\structA$ 
lacks interesting operations  also has been confirmed in both proofs. 

The universal-algebraic approach to finite-domain constraint satisfaction problems has been generalized to capture the computational complexity in many other similar settings. In particular, the complexity of $\csp(\mathbb{A})$
depends on the algebra corresponding to $\mathbb{A}$ when $\mathbb{A}$ is $\omega$-categorical~\cite{BodirskyN06}, i.e., all countable models of the first-order theory of $\structA$ are isomorphic. In particular, all structures
first-order definable 
in (reducts of) (countably infinite) homogeneous structures over finite signatures are
$\omega$-categorical structure. (A structure is homogeneous if every isomorphism between its finite substructures may be extended to an automorphism of a structure.)
Considering these infinite structures significantly broadens the class of problems that may be captured within the CSP framework. In particular, the order over rational numbers $(\Q,<)$, which  is homogeneous, gives rise to $\csp(\Q; <)$ that can be seen as the digraph acyclicity problem. The latter  cannot be expressed as the CSP over a finite template. Furthermore a number of problems of interest in qualitative reasoning may be captured by $\csp(\mathbb{A})$ where $\mathbb{A}$ is a reduct of a homogeneous 
structure $\structB$. 
It concerns constraint satisfaction problems in formalisms such as Allen's interval algebra or RCC-5, see~\cite{BodirskyJonsson-survey} for a survey. Many of the homogeneous structures $\mathbb{B}$ of interest 
are finitely bounded, i.e., there exists a finite unique minimal set $\bounds_{\mathbb{B}}$ of finite structures  over the signature of $\mathbb{B}$ such that a finite structure $\Delta$ embeds into 
$\mathbb{B}$ if and only if
none of the structures in $\bounds_{\mathbb{B}}$ embeds into $\Delta$. 
A dichotomy for algebras corresponding to reducts of 
countably infinite finitely bounded  homogeneous structures was proved in~\cite{algebraic-dichotomy-omegacat}. 
As in the finite case, it suggests the delineation between polynomial-time  solvable and NP-hard CSPs.
Although the complexity dichotomy is still far from being obtained, the algebraic dichotomy conjecture for reducts of finitely bounded homogeneous structures is known to hold in the number of cases including the reducts of $(\N,=)$~\cite{ecsps}, $(\Q,<)$~\cite{tcsps-journal}, the random partial order~\cite{poset-stacs} or a  countably infinite  homogeneous graph~\cite{BodirskyP15,equiv-csps,homographs-arxiv}.

 Theoretical research on CSPs is focused not only on providing classifications of computational complexity but also on settling the limits of applicability of widely known algorithms or algorithmic techniques such as
establishing local consistency. 
This method is used not only for finite CSP
but is also considered to be the most important (if not the only)
algorithmic technique for qualitative CSPs~\cite{Renz12}. 
The algebraic
characterization of finite structures $\mathbb{A}$ with bounded width~\cite{BoundedWidth}, i.e., for which $\csp(\mathbb{A})$ can be solved by establishing local consistency, is considered to be an important step towards solving 
the Feder-Vardi conjecture. Thus, in order to
understand the complexity of CSPs for reducts 
$\mathbb{A}$ of finitely bounded homogeneous structures, we need to 
characterize $\mathbb{A}$ with bounded width
and to understand 
how different notions of consistency relate to each other for templates under consideration.
The focus of this paper is on the latter. 

The amount of consistency needed to solve $\csp(\mathbb{A})$ for $\mathbb{A}$
with bounded width is measured here~\cite{BRWHierarchy} and here~\cite{BulatovRW} by  relational width. The relational width of $\mathbb{A}$ is  a pair of numbers $(k,l)$ with $k \leq l$ (for the exact definition we refer the reader to Section~\ref{sect:consmin}). The following question was of interest for finite structures.

\begin{question}
\label{quest:boundedrelwidth}
What is the exact relational width of $\mathbb{A}$ with bounded width?
\end{question}

Question~\ref{quest:boundedrelwidth} for finite $\mathbb{A}$ was completely answered in~\cite{BRWHierarchy} where it was proved that $\mathbb{A}$ with bounded width has always either relational width $(1,1)$ or $(2,3)$, see~\cite{BulatovRW} for another proof. Both proofs rely, however, on the algebraic characterization  of structures $\mathbb{A}$ with bounded  width. 
Although the notion of bounded width has been generalized to $\omega$-categorical  structures~\cite{Datalog-omegacat}, according to our knowledge, no algebraic characterization of bounded width for such structures is within sight.  Nevertheless the algebraic characterization of strict bounded width has been  quite easily lifted from finite~\cite{FederVardi} to infinite domains~\cite{Datalog-omegacat}. (Again, for a detailed definition we refer the reader to Section~\ref{sect:consmin}.) A reduct of a finitely bounded homogeneous structure has bounded strict width if and only if it is preserved by so-called  oligopotent quasi near-unanimity operation. This algebraic characterization gives us a hope to answer the following question analogous to Question~\ref{quest:boundedrelwidth}.

\begin{question}
\label{quest:relstrictwidth}
What is the exact relational width of reducts $\mathbb{A}$ of  finitely bounded homogeneous structures with bounded strict width?
\end{question}

In this paper we answer Question~\ref{quest:relstrictwidth} for first-order expansions of  countably infinite homogeneous graphs.
We believe that our method may be used to provide the general answer in the near future.
We note that the answer to Question~\ref{quest:relstrictwidth}
would be not only a nice theoretical result but should be also of particular interest for structures that give rise to constraint satisfaction problems in qualitative reasoning. In this context strict width is called \emph{local-to-global consistency} and has been widely studied see, e.g.,~\cite{Dechter}.

\subsection{Our results}

 In contrast to all homogeneous structures, all countably infinite homogeneous graphs are well understood and have been classified in~\cite{LachlanWoodrow}. It happens that every such graph $\homograph$ is also finitely bounded, i.e., in each case there exists a finite unique minimal set of finite graphs $\bounds_{\homograph}$ such that a finite graph $G$ embeds into $\homograph$ if and only if none of the graphs in $\bounds_{\homograph}$ embeds into $G$. 
 We will write $\maxbound_{\homograph}$ for the maximum of the number $3$ and the size of the largest finite structure in $\maxbound_{\homograph}$.
Perhaps the best known example of a homogeneous graph is the random graph that is  determined up to isomorphism
by the two properties of being homogeneous  and universal (i.e., it contains all
countable graphs as induced subgraphs).  Equivalently, the random graph is a 
unique countably
infinite graph which has this extension property:
 for all disjoint finite subsets $U, U'$ of the domain there exists an element $v$ such that $v$ is adjacent to all
members of $U$ and to none in $U'$. In this case the finite set of bounds consists of a single directed edge and a loop, and hence $\mathbb{L}_G$ for the random graph $G$ is $3$. Furthermore, the family of homogeneous graphs contains  universal countable $k$-clique free graphs $H_k$ with $k \geq 3$, called also Henson graphs, in which case $\bounds_{H_k}$ contains also a $k$-clique, and hence $\mathbb{L}_{H_k}$ is $k$  or  the graphs $C_n^s$ that are disjoint sums of $n$ cliques of size $s$ where $1 \leq n,s \leq \omega$ and either $n$ or $s$ equals $\omega$. Observe that $\bounds_{C_n^s}$
contains a graph on three vertices with two edges and one non-edge as well as a null graph over $n+1$ vertices in case $n$ is finite or a $(s+1)$-clique in the case where $s$ is finite. Thus, $\mathbb{L}_{C_n^s}$ is either $3, n+1$ or $s+1$. 
All remaining homogeneous graphs are the complements of graphs $H_k$ or $C_n^s$. 
In this paper we prove the following. 

\begin{flushleft}
\textbf{Main result.} \textit{Let $\mathbb{A}$ be a first-order expansion of a countably infinite homogeneous graph $\homograph$ such that $\mathbb{A}$ has bounded strict width. Then $\mathbb{A}$ has relational width $(2, \maxbound_{\homograph})$.}
\end{flushleft}

In fact, we obtain a more general result. Some sufficient conditions implying that a first-order expansion of a homogeneous graph $\homograph$ has relational width $(2, 
\maxbound_{\homograph})$ are given in Section~\ref{sect:suffconditions}. In particular, 
the conditions cover all languages under consideration preserved by binary canonical operations considered in~\cite{BodirskyP15,equiv-csps,homographs-arxiv} where an analysis of algebras corresponding to reducts of homogeneous graphs and the computational dichotomy is provided. Our result: relational width $(2, \maxbound_{\homograph})$ beats the upper bound $(2m, 3m)$,
where $m = \max(\textrm{arity}(\mathbb{A})+1, \maxbound, 3)$ and
$\textrm{arity}(\mathbb{A})$ is the largest arity of a relation in $\mathbb{A}$,
 that can be easily obtained from the proof of Theorem~4.10~in~\cite{UnaryDichotomy}. 

We believe that measuring relational width of structures with bounded width is interesting in its own rights. Nevertheless, our research has complexity consequences. As in the finite case, it was proved in~\cite{Datalog-omegacat} that $\csp(\mathbb{A})$ for an $\omega$-categorical $\mathbb{A}$ with strict width $k$  may be solved by establishing $(k,k+1)$-consistency and hence  in time $O(n^{k+1})$ where $n$ is the number of variables in an instance. Our main result implies that such $\csp(\mathbb{A})$ for a first-order expansion $\mathbb{A}$
of a homogeneous graph $\homograph$ may be solved by establishing $(2,  \maxbound_{\homograph})$-minimality,  and hence in time  $O(n^{m})$ 
where $m = \max(\maxbound_{\homograph}, \textrm{arity}(\mathbb{A}))$.

\subsection{Outline of the paper}

We start with general preliminaries in Section~\ref{sect:preliminaries}. Then we review canonical operations providing tractability for reducts of homogeneous graphs, Section~\ref{sect:canonicalop}. Bounded (relational) width, strict width
and other notions related to local consistency  are provided in Section~\ref{sect:consmin}. There we also give a number of examples 
explaining the applicability of our main result.
The proof of the main result is divided into Section~\ref{sect:suffconditions} and Section~\ref{sect:exactcharacterization}. In the former one,  we give a number of sufficient conditions implying relational width $(2, \maxbound_{\homograph})$, while in the latter one we show that the sufficient conditions are satisfied whenever a first-order expansion of $\homograph$ has bounded strict width. In Section~\ref{sect:suffconditions} we additionally show that the sufficient condition are also satisfied by first-order expansions of homogeneous graphs preserved by the studied binary canonical operations. As a consequence, we obtain that all tractable (whose CSP is solvable in polynomial time) reducts of $\homograph$ where $\homograph$ is $C^{\omega}_1, C^1_{\omega}, C^{\omega}_{\omega}$ or $H_k$ with $k \geq 3$ have bounded relational width $(2, \maxbound_{\homograph})$ and hence can be solved by establishing  $(2, \maxbound_{\homograph})$-minimality.

\section{Preliminaries}
\label{sect:preliminaries}

We write $t =(t[1],...,t[n])$ for a tuple of elements and $[n]$ to denote the set $\{ 1,\ldots,n \}$. 

\subsection{Relations, languages and formulas}

In this paper we consider  first-order expansions 
$\mathbb{A} := (A; E,  R_1, \ldots, R_k)$ over a finite signature $\tau$
of  homogeneous graphs, called also (constraint) languages or templates, 
where all $R_1, \ldots, R_k$ have a first-order definition in $(A;E)$.
We assume that $\mathbb{A}$ constains $=$ and $N$ whenever $N$ is pp-definable in $\mathbb{A}$.
Relations $E$ and $N$ refer always to a homogeneous graph $\homograph$ known from the context. 
For the sake of presentation we usually do not distinguish between a relation symbol $R$ in the signature of $\mathbb{A}$ and 
the relation $R^{\mathbb{A}}$ and use the former symbol for both. 
We often write $O, O_1, O_2, \ldots$ for elements of $\{ E, N, = \}$ and $\uuE, \uuN, \uuO, \uuOone, \uuOtwo$  to denote relations $(E \cup =), (N \cup =), (O \cup =), (O_1 \cup =), (O_2 \cup =)$, respectively. 

For a structure $\mathbb{A}$ over domain $A$ and a tuple $t \in A^k$, the orbit of $t$ in $\mathbb{A}$ is the relation
$\{ (\alpha(t[1]), \ldots, \alpha(t[k])) \mid \alpha \in \Aut(\mathbb{A}) \}$ where $\Aut(\mathbb{A})$
is the set of automorphisms of $\mathbb{A}$.  In particular, $E, N$ and $=$
are orbits of pairs, called also orbitals.
We would like to note that all structures considered in this paper are $\omega$-categorical.  By a theorem proved independently by Ryll-Nardzewski, Engeler and Svenonius, 
a structure $\mathbb{A}$ is $\omega$-categorical if and only if its automorphism group is oligomorphic, i.e., for every $n$ the number of orbits of $n$-tuples is finite.
See~\cite{Hodges} for a textbook on model theory.

A primitive-positive (pp-)formula is a first-order formula built exclusively out of existential quantifiers $\exists$, conjunction $\wedge$ and atomic formulas 
$R(x_1, \ldots , x_k)$ where  $R$ is a $k$-ary relation symbol and $x_1, \ldots, x_k$ are variables, not necessarily pairwise different.

\subsection{The universal-algebraic approach}

We say that an operation $f :A^n \rightarrow A$ is a polymorphism of an $m$-ary relation $R$ iff for any $m$-tuples $t_1,\ldots,t_n \in  R$, it holds that the tuple $(f(t_1[1],\ldots,t_n[1]),\ldots,f(t_1[m],\ldots,t_n[m]))$ is also in $R$. We  write $f(t_1,\ldots,t_n)$ as a shorthand for $(f(t_1[1],\ldots,t_n[1]),\ldots,f(t_1[m],\ldots,t_n[m]))$. An operation $f$ is a polymorphism of $\mathbb{A}$ if it is a polymorphism of every relation in $\mathbb{A}$. If $f : A^n \rightarrow A$ is a polymorphism of $\mathbb{A}$, $R$, we say that $f$ preserves $\mathbb{A}, R$, otherwise that $f$ violates $\mathbb{A}, R$.
A set of polymorphisms of an $\omega$-categorical structure $\mathbb{A}$ forms an algebraic object called an oligomorphic locally closed clone~\cite{OligoClone},
which in particular contains an oligomorphic permutation group~\cite{oligo}.

\begin{theorem}
(\cite{BodirskyN06}) 
\label{thm:Galoisconn}
Let $\mathbb{A}$ be a countable $\omega$-categorical structure. Then $R$ is preserved by the polymorphisms of $\mathbb{A}$ if and only if it has a  primitive-positive definition in $\mathbb{A}$, i.e., a definition via a primitive-positive formula.
\end{theorem}

We say that a set of operations $F$ generates a set of operations $G$ if every $g \in G$ is in the smallest locally-closed clone containing $F$.
We wite $\overline{\textrm{Aut}(\mathbb{A})}$ to denote the clone generated by the automorphisms of the structure $\mathbb{A}$.
An operation $f$ of an oligomorphic clone $F$ is called oligopotent if $\{ g \}$ where $g(x) := f(x, \ldots, x)$ 
is generated by the permutations in $F$.
We say that a $k$-ary operation $f$ is a weak near-unanimity operation if 
$f(y, x, \ldots, x) = f(x,y,x, \ldots, x) = \cdots = f(x, \ldots, x, y)$ for all
$x,y \in A$ and that $f$ is a quasi near-unanimity operation (short, qnu-operation) if it is a weak near-unanimity and it additionally satisfies 
$f(x, \ldots, x) = f(x, \ldots, x,y)$ for all $x,y \in A$.
We say that a $k$-ary operation $f$ is a weak near-unanimity operation modulo $\overline{\textrm{Aut}(\mathbb{A})}$ if there exist $e_1, \ldots, e_k \in \overline{\textrm{Aut}(\mathbb{A})}$ such that:
$e_1(f(y, x, \ldots, x)) = e_2(f(x,y,x, \ldots, x)) = \cdots = e_k(f(x, \ldots, x, y))$ for all
$x,y \in A$.

\subsection{The constraint satisfaction problem}
\label{sect:CSP}

 We define the CSP to be a computational problem 
 whose instance $\mathcal{I}$ is a triple $(\V, \mathcal{C}, A)$ where $\V = \{ v_1, \ldots, v_n \}$ is a set of variables, $\mathcal{C}$ is a set of constraints each of which is of the form $((v_{i_1}, \ldots, v_{i_k}), R)$ where $\{v_{i_1}, \ldots, v_{i_k} \} \subseteq \V$ is \emph{the scope} of the constraint and $R \subseteq A^k$.
 The question is whether there is a solution $\mathbf{s}: \mathcal{V} \to A$ to $\mathcal{I}$ satisfying $(\mathbf{s}(v_{i_1}), \ldots, \mathbf{s}(v_{i_k})) \in R$ for all $((v_{i_1}, \ldots, v_{i_k}), R ) \subseteq \mathcal{C}$. Further, we define 
 $\csp(\mathbb{A})$ for a constraint language $\mathbb{A}$ to be the CSP restricted to instances where all
relations come from $\mathbb{A}$.~\footnote{Equivalently, one defines an instance of $\csp(\mathbb{A})$ as a
conjunction $\varphi$ of atomic formulae over the signature of $\mathbb{A}$. Then the question is whether $\varphi$ is satisfiable in $\mathbb{A}$.}

We define  the projection of 
$((v_{i_1}, \ldots, v_{i_k}), R)$ to the set  $\{ w_1, \ldots, w_l \} \subseteq \{ v_{i_1}, \ldots, v_{i_k} \}$
to be the constraint 
$(\{ w_{1}, \ldots,  w_{l} \}, R')$ where the relation $R'$ is given by 
$(R'(w_1, \ldots, w_l) \equiv \exists x_{1} \ldots \exists x_m~R(v_{i_1}, \ldots, v_{i_k}))$ and $\{ x_1, \ldots, x_m \} = \{ v_{i_1}, \ldots, v_{i_k} \} \setminus \{ w_1, \ldots, w_m \}$.
Let $W \subseteq \V$. An assignment $\textbf{a}: W \rightarrow A$ is a partial solution to $\mathcal{I}$ if $\textbf{a}$ satisfies all projections of constraints in $\mathcal{I}$ to variables in $W$.

It is very well known that adding pp-definable relations to the template does not change the complexity of the problem. 

\begin{proposition}
\label{prop:GaloisComp}
Let $\mathbb{A} = (A; R_1, \ldots,R_l)$ be a relational structure, and let $R$ be a relation that has a primitive-positive definition in $\structA$. Then $\csp(\mathbb{A})$ and $\csp(A, R,R_1,\ldots,R_l)$ are log-space equivalent.
\end{proposition}

\subsection{Efficient entailment}
\label{sect:effent}

We say that a formula $\varphi_1$ entails a formula $\varphi_2$ both over free variables $x_1, \ldots, x_n$
if $(\forall x_1 \cdots \forall x_n~(\varphi_1(x_1,\ldots, x_n) \implies \varphi_2(x_1, \ldots, x_n)))$ is a valid 
sentence. Furthermore, we say that an $n$-ary relation $R$ entails $\varphi$ over free variables $x_1, \ldots, x_n$
if $R(x_1, \ldots, x_n)$ entails $\varphi$. These definitions are quite standard but for the purposes of this paper we need a stronger 
notion of entailment.

\begin{definition}
\label{def:effentailment}
We say that a quaternary relation $R$ efficiently entails
$\psi := (S_1(x_1, x_2) \implies S_2(x_3, x_4))$ where $S_1, S_2$ are binary relations 
if $R$ entails $\psi$ and 
$R$  contains  
\begin{enumerate}
\item \label{effentailment:forwardtuple} 
a tuple $t_1$ 
such that $(t_1[1],t_1[2]) \in S_1$ and $(t_1[3], t_1[4]) \in S_2$, and 
\item \label{effentailment:backwardtuple} 
a tuple $t_2$ 
such that $(t_2[1],t_2[2]) \notin S_1$ and $(t_2[3], t_2[4]) \notin S_2$.
\end{enumerate} 
\end{definition}

We say that a quaternary relation $R$ is a 
$[(S_1(x_1, x_2) \implies S_2(x_3, x_4)), (\varphi)]$-relation
if $R$ efficiently entails $(S_1(x_1, x_2) \implies S_2(x_3, x_4))$ and entails $\varphi$ or that 
a quaternary relation $R$ is a 
$[(S_1(x_1, x_2) \implies S_2(x_3, x_4))]$-relation
if $R$ efficiently entails $(S_1(x_1, x_2) \implies S_2(x_3, x_4))$.


\section{Canonical Operations over Reducts of Homogeneous Graphs}
\label{sect:canonicalop}

The polymorphisms that appear in the complexity classifications of CSPs of reducts of homogeneous graphs display some regularities in the sense defined below.  

Let $f : A^k \to A$, and let $G$ be a permutation group on $A$.
We say that $f$ is
canonical with respect to $G$  if for all $m \in N, \alpha_1, \ldots, \alpha_k \in G$ and m-tuples 
$a_1,\ldots,a_k$, there exists $\beta \in G$ 
such that $\beta f(\alpha_1 (a_1), \ldots , \alpha_k (a_k)) = f(a_1, \ldots , a_k)$.
 Equivalently, this means that $f$ induces an
operation $\xi^\textrm{typ}(f)$, called a $k$-ary behaviour, on orbits of m-tuples under $G$, by defining $\xi^{\textrm{typ}}(f)(O_1,\ldots,O_k )$ as the orbit of $f(a_1, \ldots, a_k)$ where $a_i$ is any m-tuple in $O_i$. 
In what follows we are mainly interested in operations that are canonical with respect to $\Aut(\homograph)$ where $\homograph$ is a homogeneous graph. Therefore we usually say simply canonical.
See~\cite{BP-reductsRamsey} for a survey on canonical operations.
Three simple binary behaviors  are presented in Figure~\ref{fig:binarybehaviors}.
According to Definition~\ref{def:minimality}, a binary injection $f$ 
such that $\xi^\textrm{typ}(f)$ is 
\begin{itemize}
\item $B_1$ is said to be of behavior max and balanced,
\item $B_2$ is said to be $E$-constant,
\item $B_3$ is said to be of type min and $N$-dominated. 
\end{itemize}

\begin{figure}
\begin{center}
\begin{tabular}{l||l|l|l|}
$B_1$ & = & E & N \\
\hline \hline
= & = & E & N \\
\hline
E & E & E & E \\
\hline
N & N & E & N \\
\hline
\end{tabular}
\quad
\begin{tabular}{l||l|l|l|}
$B_2$ & = & E & N \\
\hline \hline
= & = & E & E \\
\hline
E & E & E & E \\
\hline
N & E & E & E \\
\hline
\end{tabular}
\quad 
\begin{tabular}{l||l|l|l|}
$B_3$ &  = & E & N \\
\hline \hline
= & = & N & N \\
\hline
E & N & E & N \\
\hline
N & N & N & N \\
\hline
\end{tabular}
\end{center}
\caption{Three examplary binary behaviours: $B_1$, $B_2$, and $B_3$.}
\label{fig:binarybehaviors}
\end{figure}

We introduce the following notation. Let $R_1, \ldots , R_k \subseteq A^2$ be binary relations. 
We write $R_1 \cdots R_k$ for the binary relation on $A^k$ defined so that: 
$R_1 \cdots R_k(a_1, a_2)$ holds for $k$-tuples $a_1,a_2 \in A^k$ if and only if 
$R_i(a_1[i] , a_2[i])$ holds for all $i \in [k]$. 
Here, we can find the list of all binary behaviours of interest.
\begin{definition}
Let $(A,E)$ be a countably infinite homogeneous graph. We say that a binary injective operation $f : A^2 \rightarrow A$ is
\begin{itemize}
\item \emph{balanced} if for all $a,b \in A^2$ we have that $\EEq(a,b)$ and $\EqE(a,b)$ implies $E(f(a),f(b))$ 
as well as $\NEq(a,b)$ and $\EqN(a,b)$ implies $N(f(a),f(b))$,
\item $E$-dominated ($N$-dominated) if for all $a,b \in A^2$ with $\NeqEq(a,b)$ or $\EqNeq(a,b)$
we have that $E(f(a),f(b))$ ($N(f(a),f(b))$);
\item of behaviour \emph{min} if for all $a,b \in A^2$ with $\NeqNeq(a,b)$ we have $E(f (a), f (b))$ iff
 $\EE(a,b)$;
\item of behaviour \emph{max} if for all $a,b \in A^2$ with $\NeqNeq(a,b)$ we have $N(f (a), f (b))$ iff $\NN(a,b)$;
\item of behaviour \emph{projection} if there exists $i \in [2]$ such that for all $a,b \in A^2$ 
with $\NeqNeq(a,b)$ we have $E(f (a), f (b))$ iff $E(a[i], b[i])$,
\item of behaviour \emph{xor} if for all  $a,b \in A^2$ with $\NeqNeq(a,b)$ the relation $E(f(a),f(b))$ holds iff $\EN(a, b)$ or $\NE(a, b)$ holds; 
\item of behaviour \emph{xnor} if for all  $a,b \in A^2$ with $\NeqNeq(a,b)$ the relation $E(f(a),f(b))$ holds iff 
$\EE(a, b)$ or $\NN(a, b)$ holds;
\item \emph{E-constant} if the image of $f$ is a clique,
\item \emph{N-constant} if the image of $f$ is an independent set.
\end{itemize}
\end{definition}

\noindent
We now turn to ternary behaviours of interest.
\begin{definition}
Let $(A; E)$ be a countably infinite homogeneous graph. We say that a ternary injective operation $f : A^3 \rightarrow A$ is of behaviour
\begin{itemize}
\item \emph{majority} if for all $a,b \in D^3$ satisfying $\neq\neq\neq(a,b)$ we have that $E(f(a), f(b))$
if and only if $EEE(a,b), EEN(a,b), ENE(a,b)$, or $NEE(a,b)$,
\item \emph{minority}  if for all $a,b \in D^3$ satisfying $\neq\neq\neq(a,b)$ we have that $N(f(a), f(b))$
if and only if $NNN(a,b), EEN(a,b), ENE(a,b)$, or $NEE(a,b)$.
\end{itemize}
Furthermore, let $B$ be a binary behavior. A ternary function is hyperplanely of behaviour $B$
if the binary functions $(x,y) \rightarrow f(x,y,c), (x,z) \rightarrow f(x,c,z)$, and $(y,z) \rightarrow f(c,y,z)$
have behavior $B$ for all $c \in D$.
\end{definition}

\section{Consistency and Minimality}
\label{sect:consmin}

This section is devoted to the formal introduction of consistency and width notions. The main algorithm we are interested in is based on establishing minimality.

\begin{definition}
\label{def:minimality}
Let $l \geq k > 0$ be natural numbers. An instance $\mathcal{I} = (\V, \mathcal{C}, A)$ of the CSP is 
$(k,l)$-minimal if:
  \begin{enumerate}
  \item \label{M1} Every at most $l$-element set of variables is contained in the scope of some constraint in $\mathcal{I}$. 
  \item \label{M2} For every set $W$ with $\left| W \right| \leq k$ and every pair of constraints $C_1$ and $C_2$ in $\mathcal{C}$ whose scopes contain $W$, the projections
  of the constraints $C_1$ and $C_2$ to $W$ are the same.  
 \end{enumerate} 
We say that $\mathcal{I}$ is trivial if it contains a constraint with an empty relation. Otherwise, we say that $\mathcal{I}$ is non-trivial.
\end{definition}

As in the finite case, one may transform 
an instance $\mathcal{I}$ into an equivalent instance, i.e. with the same set of solutions by simply introducing at most $O(\left| \V \right|^l)$ new constraints so that the first condition in Definition~\ref{def:minimality} was satisfied and then 
by repeatedly removing orbits of tuples from constraints until the second condition is satisfied. Similarly to the finite CSP we have the following.

\begin{proposition}
\label{prop:minimality}
Let $\mathbb{A}$ be an $\omega$-categorical relational structure. 
Then for every instance $\mathcal{I}$ of $\csp(\mathbb{A})$ and $l \geq k > 0$ there exists an instance $\mathcal{I}'$ of the CSP  with the same sets of solutions as $\mathcal{I}$ 
such that $\mathcal{I}'$ is $(k,l)$-minimal. 

For fixed $(k,l)$ and $\mathbb{A}$, the process of establishing $(k,l)$-minimality, i.e., transfoming $\mathcal{I}$ into $\mathcal{I}'$ takes time $O(\left| \V \right|^m)$ where $m = \max(l, \textrm{arity}(\structA))$ is the maximum of $l$ and the greatest arity of a relation in $\structA$. If $\mathcal{I}'$ is trivial, then both $\mathcal{I}$ and $\mathcal{I}'$ 
have no solutions.
\end{proposition}
\noindent
 We are now ready to define the relational width.
\begin{definition}
We say that $\mathbb{A}$ has relational width $(k,l)$ if and only if 
$\mathcal{I}$ has a solution provided any $(k,l)$-minimal instance of the CSP
equivalent to $I$ is non-trivial. 
We say that $\mathbb{A}$ has (relational) 
bounded width~\footnote{We note that the definition of bounded width provided in~\cite{Datalog-omegacat} is equivalent to ours.}  if there exist $(k,l)$ such that $\mathbb{A}$ has relational width $(k,l)$.
\end{definition}



   Finite structures with bounded relational width admit an algebraic characterization~\cite{BoundedWidth}. 
It is known that a finite structure $\structA$ 
has bounded (relational) width if and only if it has a four-ary polymorphism $f$ and a ternary polymorphism $g$ that are weak near-unanimity operations and such that $f(y,x,x,x) = g(y,x,x)$ for all $x, y \in A$.
We have a similar sufficient condition for reducts of finitely bounded homogeneous structures.

\begin{theorem}
\label{thm:BodMot}
(\cite{UnaryDichotomy})
Let $\mathbb{A}$ be a finite-signature reduct of a finitely bounded homogeneous structure $\mathbb{B}$. Suppose that $\mathbb{A}$ has a four-ary polymorphism $f$ and a ternary polymorphism $g$ that are canonical with respect to $\Aut(\mathbb{B})$ and are weak near-unanimity operations modulo $\overline{\Aut(\mathbb{B})}$, and such that there are operations $e_1,e_2 \in \overline{\Aut(\mathbb{B})}$ with $e_1(f(y,x,x,x)) = e_2(g(y,x,x))$ for all $x, y \in A$. Then $\csp(\mathbb{A})$ has bounded relational width.
\end{theorem} 

A slight change in the proof of the above theorem gives us the upper bound  for relational width of infinite structures under consideration.

\begin{corollary}
\label{cor:relwidth}
Let $\mathbb{A}$ be a finite-signature reduct of a finitely bounded homogeneous structure $\mathbb{B}$. Suppose that $\mathbb{A}$ has a four-ary polymorphism $f$ and a ternary polymorphism $g$ that are canonical with respect to $\Aut(\mathbb{B})$, that are weak near-unanimity operations modulo $\overline{\Aut(\mathbb{A})}$, and such that there are operations $e_1,e_2 \in \overline{\Aut(\mathbb{B})}$ with $e_1(f(y,x,x,x)) = e_2(g(y,x,x))$ for all $x, y \in A$. Then $\mathbb{A}$ has 
relational width $(2m, 3m)$ where $m = \max(\textrm{arity}(\mathbb{A})+1,
\textrm{arity}(\mathbb{B})+1, \mathbb{L}_{\mathbb{B}}, 3)$.
\end{corollary}

We now use Corollary~\ref{cor:relwidth}  to provide the upper bound of the relational width for reducts of homogeneous graphs preserved by binary canonical operations considered in~\cite{BodirskyP15, equiv-csps,homographs-arxiv}.

\begin{proposition}
\label{prop:binRW}
Let  $\mathbb{A}$ be a reduct of a countably infinite homogeneous graph $\homograph$ preserved 
by a binary injection:
\begin{enumerate}
\item  of behaviour max which is either balanced or $E$-dominated, or 
\item   of  behaviour min which is either balanced or $N$-dominated, or 
\item   which is  $E$-constant, or
\item   which is  $N$-constant.
\end{enumerate}
Then it has relational width $(2m, 3m)$ where $m = \max(\textrm{arity}(\mathbb{A})+1, \maxbound_{\homograph}, 3)$.
\end{proposition}

In Section~\ref{sect:binaryinjections},  we use our approach to show that the exact relational width of
structures  under consideration in Proposition~\ref{prop:binRW}
 is $(2, \maxbound_{\homograph})$. The same is proved 
 for first-order expansions of homogeneous graphs with bounded strict width.

Strict width is defined as follows. A $(k,l)$-minimal instance $\mathcal{I}$ of the CSP 
is called globally consistent, if every partial solution of $\mathcal{I}$  can be extended to 
a total solution of~$\mathcal{I}$. 

\begin{definition}
We say that $\mathbb{A}$ has strict width $k$ if
 for some $l \geq k \geq 2$ all instances of $\csp(\mathbb{A})$ that are $(k,l)$-minimal are globally consistent. We say that $\mathbb{A}$
has bounded strict width if it has strict width $k$ for some $k$.~\footnote{Our definition of strict width slightly varies from a definition in~\cite{Datalog-omegacat}  but again both definitions are equivalent.} 
\end{definition}
We have the following algebraic  characterization of $\omega$-categorical structures with bounded strict width.

\begin{theorem}
\cite{Datalog-omegacat,OligoClone}
\label{thm:strictwidth}
Let $\mathbb{A}$ be an $\omega$-categorical language. Then  the following are equivalent.
\begin{enumerate}
 \item $\mathbb{A}$ has strict width $k$.
 \item $\mathbb{A}$ has an oligopotent $(k+1)$-ary quasi near-unanimity operation as a polymorphism.
\end{enumerate}
\end{theorem}

For a $(2,k)$-minimal instance over variables $\V = \{ v_1,\ldots, v_n \}$ we write $\instance_{i,j}$ with $i,j \in [n]$ to denote a subset of $\{ E, N, = \}$ such that the projection of all constraints having $v_i, v_j$ in its scope to $\{ v_i, v_j \}$
equals $\bigcup_{O \in \instance_{i,j}} O$. We will say that an instance is \emph{simple}
if $\left| \instance_{i,j} \right| = 1$ for all $i,j \in [n]$.

We will now show that a simple non-trivial $(2, \maxbound_{\homograph})$-instance of $\csp(\structA)$ for a first-order  expansion $\structA$ of a homogeneous graph $\homograph$ always has a solution and that this amount of consistency is necessary.


\begin{observation}
\label{obs:ENEqualityWidth}
Let $\mathcal{I}$ be a simple non-trivial $(2, \maxbound_{\homograph})$-minimal instance of the CSP equivalent to an instance 
of $\csp(\homograph')$ where $\homograph'$ is the expansion of $\homograph$ containing all orbitals pp-definable in $\homograph$.
Then $\mathcal{I}$ has a solution.

On the other hand, for every  homogeneous graph $\homograph$ there exists a simple non-trivial $(1, \maxbound_{\homograph})$-minimal instance $\mathcal{I}_1$ equivalent to an instance of $\csp(\homograph')$
and a simple non-trivial $(2, \maxbound_{\homograph} - 1)$-minimal instance $\mathcal{I}_2$ 
equivalent to an instance of  $\csp(\homograph')$  that have no solutions.
\end{observation}

\begin{proof}
We start from proving the first part of the observation.
Define $\Delta$ to be a finite structure over the domain consisting of variables $\{ v_1, \ldots, v_n \}$  in $\instance$ and the signature $\tau  \subseteq \{ E, N, = \}$ such that $(v_i, v_j) \in R^{\Delta}$ for $i,j \in [n]$ and $R \in \tau$ if $\varphi_I$ contains a constraint $((v_i, v_j), R)$.
Since $\instance$ is $(2,3)$-minimal, we have the following.
\begin{observation}
\label{obs:minequality}
The binary relation $\sim := \{ (v_i, v_j) \mid \instance_{i,j} \subseteq \{ = \} \}\cup \bigcup_{i \in [n]} \{ (v_i v_i) \}$
is an equivalence relation.
\end{observation}

We claim that there is an embedding from $\Delta/\sim$ to $\homograph'$. Assume the contrary. Since $\homograph$ is finitely bounded, there exists 
$G$ over variables $\{ w_1, \ldots, w_l \}$ in $\bounds_{\homograph'}$ such that $G$ embeds into $\Delta/\sim$ and $l \leq \maxbound_{\homograph}$. Since $\mathcal{I}$ is $(2, \maxbound_{\homograph})$-minimal, there is a constraint $C$ in $\mathcal{I}$  whose scope contains $\{ w_1, \ldots, w_l \}$ and the corresponding relation is empty. It contradicts with the assumption that $\instance$ is non-trivial. Thus, $\Delta/\sim$ embeds into $\homograph$, and in consequence $\instance$ has a solution. It completes the proof of the first part of the observation.

For the second part of the observation, we select  $\mathcal{I}_1$ to be  $\{ ((v_1,v_2),E), ((v_1, v_2), N) \}$.
Indeed, every subset of variables of $\mathcal{I}_1$ is in the scope of some constraint. The projection of each constraint to $\{ v _1 \}$ or $\{ v_2 \}$ is the set of all vertices in $\homograph$.
It follows that $\mathcal{I}_1$ is $(1, \maxbound_{\homograph})$-consistent. Clearly $\mathcal{I}_1$ has no solutions.

We now turn to $\mathcal{I}_2$. If $\maxbound_{\homograph} > 3$ and $G = ([n], E)$ is a forbidden subgraph of size $n = \maxbound_{\homograph}$ consider
an instance $\mathcal{I}'_2$ over variables $\{ v_1, \ldots, v_n \}$ 
containing a  constraint $((v_i, v_j), E)$ if $(i,j) \in E^G$ and   a constraint 
$((v_i, v_j), N)$ if $(i,j) \notin E^G$.
Let $\mathcal{I}_2$ be a $(2, \maxbound_{\homograph}- 1)$-minimal instance of the CSP
equivalent to $\mathcal{I}'_2$. By the minimality of $\bounds_{\homograph}$, we have that no induced subgraph of $G$ is in $\bounds_{\homograph}$. It follows that $\mathcal{I}_2$ is non-trivial but, clearly, 
$\mathcal{I}_2$  has no solutions. If $\maxbound_{\homograph} = 3$, then we select $\mathcal{I}_2$ to be an instance such that $\mathcal{C} = \{ ((v_1,v_2), =),
((v_2, v_3), =),  ((v_1, v_3), E)$. It is again straightforward to check that 
$\instance_2$ is $(2,2)$-minimal. Yet, it has no solutions. It completes the proof of the observation.
\end{proof}

We complete this section by giving  some examples of first-order expansions of homogeneous graphs with bounded strict width.

\begin{proposition}
\label{prop:ClausesQNU}
Let $\mathbb{A}$ be a first-order expansion of the random graph $\homograph = (A; E)$ 
such that every relation in $\mathbb{A}$
is pp-definable as a conjunction of clauses of the form:
\begin{align}
(x_1 \neq y_1 \vee \cdots \vee x_k \neq y_1 \vee R(y_1, y_2) \vee y_2 \neq z_1 \vee \cdots \vee y_2 \neq z_l), \nonumber
\end{align}
where $R \in \{ E, N \}$.
Then $\mathbb{A}$ has bounded strict-width.  
\end{proposition}

\noindent
And here comes another example.

\begin{proposition}
\label{prop:qnuExTwoEquivClass}
The constraint language $\mathbb{A} = (A; E, N, R)$ where 
$(A; E)$ is $C^{\omega}_2$ and $R(x_1, x_2, x_3) \equiv ((E(x_1, x_2) \wedge N(x_2, x_3)) \vee (N(x_1, x_2) \wedge E(x_2, x_3)))$ has bounded strict width. 
\end{proposition}

\section{Conditions Sufficient for Low Relational Width}
\label{sect:suffconditions}

In order to show  that a non-trivial $(2, \maxbound_{\homograph})$-minimal instance $\instance$ of $\csp(\structA)$ for first-order expansions $\structA$ of a homogeneous graph $\homograph$ has a solution, we always use one scheme. We take advantage of the fact that certain  quaternary and ternary relations are not pp-definable in $\structA$ and we carefully 
narrow down $\instance_{i,j}$ for $i,j \in [n]$ 
so that we end up with a  simple non-trivial instance $\instance'$ which is a 'subinstance' of $\instance$ in the following sense: for every $C = ((x_1, \ldots, x_r), R)$ in $\instance$ we have $C' = ((x_1, \ldots, x_r), R') \in \instance'$ such that $R' \subseteq R$. Since $\instance'$ is simple, by Observation~\ref{obs:ENEqualityWidth}, it has a solution. This solution is clearly a solution to the orginal instance $\instance$.

We shrink an instance of $\csp(\structA)$ using one of three different sets of relations presented in the three lemmas below.

\begin{lemma}
\label{lem:injO1dominating}
Let $\{ O_1, O_2 \}$ be $\{ E, N \}$ and $\mathbb{A}$ be a first-order expansion of a  homogeneous graph $\homograph$  such that none of the following types of relations is pp-definable in $\mathbb{A}$:
\begin{enumerate}
\item \label{injO1dom:O1uuO2} $[(O_1(x_1, x_2) \implies \uuOtwo(x_3, x_4))]$-relations,
\item \label{injO1dom:O1Eq} $[(O_1(x_1, x_2) \implies x_3 = x_4)]$-relations, and
\item \label{injO1dom:O2Eq} $[(O_2(x_1, x_2) \implies x_3 = x_4), (\uuOtwo(x_1, x_2) \wedge \uuOtwo(x_3, x_4))]$-relations.
\end{enumerate}
Then  $\mathbb{A}$ has relational width $(2, \maxbound_{\homograph})$.
\end{lemma}

Before we discuss the 'shrinking' strategy that stands behind Lemma~\ref{lem:injO1dominating}, consider a non-trivial  instance $\instance$ of some $\structA$ under consideration in the lemma and a constraint $((x_1, \ldots, x_r), R)$
for which there are $i_1, j_1, i_2, j_2$ such that $v_{i_1}, v_{j_1}, v_{i_2}, v_{j_2} \in \{ x_1, \ldots, x_r \}$ and $O_{1} \in \instance_{i_1, j_1}$, $O_1 \in \instance_{i_2, j_2}$. Since $\instance$ is non-trivial and $(2,3)$-minimal the relation $(R'(x_1, \ldots, x_r) \equiv (R(x_1, \ldots, x_r) \wedge O_1(v_{i_1}, v_{j_1})))$ is non-empty.
But also  
$(R''(x_1, \ldots, x_r) \equiv (R(x_1, \ldots, x_r) \wedge O_1(v_{i_1}, v_{j_1}) \wedge O_1(v_{i_2}, v_{j_2})))$ is non-empty. Indeed,  otherwise since $O_1 \cup \uuOtwo = A^2$, the structure $\structA$ would define a relation from Item~\ref{injO1dom:O1uuO2} or Item~\ref{injO1dom:O1Eq}. Generalizing 
the argument, one can easily transform $\instance$ 
to a non-trivial $\instance'$ where
every $\instance'_{i,j} = \{ O_1 \}$ whenever $\instance_{i,j}$ contains $O_1$. Using a similar reasoning and Item~\ref{injO1dom:O2Eq}, and taking care of some details, we have to skip here, one can then transform $\instance'$ to $\instance''$ so that $\instance''_{i,j} = O_2$ whenever $\instance'_{i,j}$ contains 
$O_2$. Since $\instance''$ is simple and non-trivial, we can use Observation~\ref{obs:ENEqualityWidth} to argue that both $\instance''$ and $\instance$ has a solution.

The next lemma considers a specific situation where $\homograph$ is a disjoint sum of $\omega$ edges and  languages under consideration are preserved by oligopotent qnu-operations.

\begin{lemma}
\label{lem:hnear23min}
Let $\mathbb{A}$ be a first-order expansion of $C^2_{\omega}$ preserved by an oligopotent qnu-operation and such that none of the following types of relations is pp-definable in $\mathbb{A}$:
\begin{enumerate}
\item \label{hnear23min:NuuE} $[(N(x_1, x_2) \implies \uuE(x_3, x_4))]$-relations,
\item \label{hnear23min:NE} $[(N(x_1, x_2) \implies E(x_3, x_4)), (\uuE(x_3, x_4)) ]$-relations,
\item \label{hnear23min:NEq} $[(N(x_1, x_2) \implies x_3 = x_4)]$-relations,
\item \label{hnear23min:O1NO2} $[(O_1(x_ 1, x_2) \implies O_2(x_3, x_4)), (\uuE(x_1, x_2) \wedge N(x_2, x_3) \wedge \uuE(x_3, x_4))]$-relations 
where the set $\{ O_1, O_2 \}$ equals $\{ E, = \}$.
\end{enumerate}
Then  $\mathbb{A}$ has relational width $(2, 3)$.
\end{lemma}

The shrinking strategy for Lemma~\ref{lem:hnear23min} is as follows. We start with a non-trivial $(2,3)$-minimal instance $\instance$ and use Items~\ref{hnear23min:NuuE}--\ref{hnear23min:NEq} to transform it into 
a non-trivial $(2,3)$-minimal $\instance'$ such that  $\instance'_{i,j} = \{N \}$ whenever $\instance_{i,j}$ contains $N$.
Since $\uuE$ fo-definable in $C^2_{\omega}$ is transitive and $\instance'$ is $(2,3)$-minimal, it is easy to show that the graph over variables $\{ v_1, \ldots, v_n\}$ and edges $\instance'_{i,j}$ with $i,j \in [n]$ is a disjoint union of components $K_1, \ldots, K_{\kappa}$ such that for all $k \in [\kappa]$ and all $v_i, v_j \in K_k$ it holds that $\instance'_{i,j} \subseteq \{ E, = \}$ and whenever $v_i, v_j$ are in different components, then $\instance'_{i,j} = \{ N \}$. Now, any $\instance'_{K_i}$ --- the instance $\instance'$ restricted to variables in $K_i$, which is in fact an instance of $\csp(\Delta)$ for $\Delta$ over two-elements (some edge in $C^2_{\omega}$, different for every $i \in [\kappa]$) preserved by a near-unanimity operation, is shown to have a solution $\solution_{i}$. It follows by the characterization of relational width for finite structure.
In order to prove that solution $\solution:= \bigcup_{i \in [\kappa]} \solution_i$ is the solution to $\instance'$ and hence to $\instance$ we use the fact that relations from Item~\ref{hnear23min:O1NO2} are not pp-definable in $\structA$.

Finally, we turn to the case where $\homograph$ is a disjoint sum of two infinite cliques and the structures $\mathbb{A}$ have  oligopotent qnu-operations as polymorphisms.

\begin{lemma}
\label{lem:Comega23excludes}
Let $\mathbb{A}$ be a first-order expansion of $C^\omega_2$ preserved by an oligopotent qnu-operation and such that $\mathbb{A}$  pp-defines neither $\uuN$ nor $[(O(x_1, x_2) \rightarrow  x_3 = x_4)]$ for any $O \in \{ E, N \}$.
Then $\mathbb{A}$ has relational width $(2,3)$.
\end{lemma}

Clearly any tuple over  $C^{\omega}_2$ takes some of its values from one equivalence class in $C^{\omega}_2$ and the remaining values from the other class. In order to prove Lemma~\ref{lem:Comega23excludes}, we  consider a non-trivial $(2,3)$-minimal instance $\instance$ of $\csp(\structA)$ but this time we also assume without loss of generality that there are no $i,j$ with $\instance_{i,j} = \{ = \}$. Since $\structA$ does not define $\uuN$
we have that for all $i,j \in [n]$ the set $\instance_{i,j}$ contains $E$.
Then we transform $\instance$ to
$\instance_B$ of $\csp(\Delta)$ where $\Delta$ is over the   domain $\{0,1 \}$ by replacing any tuple $t$ in any relation in any constraint in $\instance$  by a tuple over $\{0,1 \}$ so that all 
values in one equivalence class are replaced by $0$ and all values in the other equivalence class are replaced by $1$. Since $\Delta$ is preserved by a near-unanimity operation, and hence has bounded relational width 
we have that the $(2,3)$-minimal $\instance_B$ 
has a solution $\solution_B: \{v_1, \ldots, v_m \} \to \{ 0,1 \}$.
We use $\solution_B$ to transform $\instance$ to $\instance'$ so that we set $\instance'_{i,j}$ to $\{ N \}$ whenever $\solution_B(v_i) \neq \solution_B(v_j)$.  No $[(N(x_1, x_2) \rightarrow x_3 = x_4)]$-relations are pp-definable in $\structA$, and hence we have that $\instance'_{i,j}$ for any $i,j \in [n]$ contains $E$. Since $\structB$ pp-defines no $[(E(x_1, x_2) \rightarrow x_3 = x_4)]$-relations we may transform $\instance'$ into $\instance''$ so that $\instance''_{i,j}$ is $E$ 
whenever $\instance'_{i,j} \neq \{ N \}$. Thus, $\instance''$ 
is a simple non-trivial $(2,3)$-minimal instance. It follows by Observation~\ref{obs:ENEqualityWidth} that both $\instance''$ and $\instance$ have a solution. Again, we skipped many details but our goal was rather to convey some intuitions that stand behind the proofs of the lemmas in this section. 

\section{Constraint Languages with Low Relational Width}
\label{sect:exactcharacterization}

In this section we employ lemmas from Section~\ref{sect:suffconditions} to provide the exact characterization of relational width of first-order expansions of homogeneous graphs with bounded strict width and first-order expansions of homogeneous graphs preserved by binary canonical operations from Proposition ~\ref{prop:binRW}. In order to prove the former, we also 
show which quaternary relations of interest are violated by ternary injections used in the complexity classification (see~Subsection~\ref{sect:ternaryinjections}). To rule out some other relations,
we have to use oligopotent qnu-operations directly (see~Subsection~\ref{sect:violatedbyQNU}). 


\subsection{Binary Injections and Low Relational Width}
\label{sect:binaryinjections}

We start with first-order expansions $\mathbb{A}$ of homogeneous  graphs $\homograph$
whose tractability has been shown in Proposition~8.22 in~\cite{BodirskyP15}, Proposition 6.2~in~\cite{equiv-csps} as well as 
Proposition~37 and Theorem~62 in~\cite{homographs-arxiv}.

\begin{lemma}	
\label{lem:graphreductbinary}
Let  $\mathbb{A}$ be a first-order expansion of a countably infinite homogeneous  graph $\homograph$ preserved 
by a binary injection:
\begin{enumerate}
\item  \label{graphreductbinary:max} of behaviour max which is either balanced or $E$-dominated, or 
\item  \label{graphreductbinary:min} of behaviour min which is either balanced or $N$-dominated, or 
\item  \label{graphreductbinary:Econstant} which is  $E$-constant, or
\item  \label{graphreductbinary:Nconstant} which is  $N$-constant.
\end{enumerate}
Then $\mathbb{A}$ has relational width $(2,\maxbound_{\homograph})$.
\end{lemma}

The above lemma gives an opportunity to reformulate the dichotomy results for reducts $\structA$ of $C^{1}_{\omega},C^{\omega}_1, C^{\omega}_{\omega}$ and $H_k$ for any $k \geq 3$.

\begin{corollary}
\label{cor:newDich}
Let $\structA$ be 
a reduct of a homomorphism graph $\homograph$ which is $C^{1}_{\omega},C^{\omega}_1, C^{\omega}_{\omega}$ or $H_k$ for any $k \geq 3$. Then either $\csp(\structA)$ is NP-complete or $\structA$ has relational width $(2,\maxbound_{\homograph})$.
\end{corollary}

\begin{proof}
We have that any tractable first-order expansion of $(\N; =, \neq)$ is preserved by a binary injection~\cite{ecsps}. It follows that every tractable reduct of $C^{\omega}_1$ is either  preserved by a constant operation or is a first-order expansion of $C^{\omega}_1$ and preserved by a binary injection which is of behaviour max and $E$-dominated. A similar reasoning holds for reducts of $C^1_{\omega}$ with a difference that we replace $E$ with $N$.
Further, every reduct of $C^{\omega}_{\omega}$ is either homomorphically equivalent to a reduct of $(\N; =)$ or pp-defines both $E$ and $N$, see Theorem~4.5 \cite{equiv-csps}. In the former case we are done, while in the latter a tractable $\structA$ is preserved by a binary injection of behaviour min and balanced,~Corollary~7.5 in~\cite{equiv-csps}. The corollary follows by 
 Lemma~\ref{lem:graphreductbinary}. By Proposition~15 and Lemma~17 in~\cite{homographs-arxiv}, a tractable reduct of $H_k$ with $k \geq 3$ is either homomorphically equivalent to a reduct of $(\N; =)$  or pp-defines both $E$ and $N$. In the former case we are done while in the latter, we have that 
$\csp(\structA)$ is in P when it is preserved by a binary injection of behaviour min and $N$-dominated (see Theorem~38 in~\cite{homographs-arxiv}).
Again, the corollary follows by Lemma~\ref{lem:graphreductbinary}.
\end{proof}

\subsection{Types of Relations violated by Ternary Canonical Operations}
\label{sect:ternaryinjections}

Here we look at quaternary relations of interest violated by canonical ternary operations. We start with ternary injections of behaviour majority.

\begin{lemma}
\label{lem:graphmajorityequalities}
Let $\mathbb{A}$ be a reduct  of a countably infinite homogeneous  graph
preserved by a ternary injection of behaviour majority which additionally is:
\begin{itemize}
\item hyperplanely balanced and of behaviour projection, or
\item hyperplanely E-constant, or hyperplanely N-constant, or
\item hyperplanely of behaviour max and E-dominated, or
\item hyperplanely of behaviour min and N-dominated.
\end{itemize}
Then $\mathbb{A}$
pp-defines no $[(O(x_1, x_2) \implies x_3 = x_4)]$-relations with $O \in \{ E, N \}$.
\end{lemma}

\noindent
We continue with ternary injections of behaviour minority.

\begin{lemma}
\label{lem:graphminorityequalities}
Let $\mathbb{A}$ be a reduct  of a countably infinite homogeneous graph
preserved by a ternary injection of behaviour minority which additionally is: 
\begin{itemize}
\item hyperplanely balanced and of behaviour projection,
\item hyperplanely of behaviour projection and E-dominated, or
\item  hyperplanely of behaviour projection and N-dominated, or
\item hyperplanely balanced of behaviour xnor,  or
\item hyperplanely balanced of behaviour xor.
\end{itemize}
Then $\mathbb{A}$ does not
pp-define $[(O(x_1, x_2) \implies x_3 = x_4)]$-relations with $O \in \{ E, N \}$.
\end{lemma}

We are already in the position to prove that another large family of $\csp(\mathbb{A})$ under consideration may be solved 
by establishing minimality.

\begin{corollary}
\label{cor:Comega23min}
Let $\mathbb{A}$ be a first-order expansion of $C^{\omega}_2$ preserved by a canonical ternary injection of behaviour minority which is hyperplanely balanced of behaviour
xnor and an oligopotent qnu-operation. Then $\mathbb{A}$ has relational width $(2,3)$. 
\end{corollary}

\begin{proof}
By Lemma~\ref{lem:graphminorityequalities}, the structure $\mathbb{A}$ does not pp-define $[(O(x_1, x_2) \implies x_3 = x_4)]$ with $O \in \{ E, N \}$.  Since a canonical ternary injection of behaviour minority which is hyperplanely balanced of behaviour
xnor does not preserve $\uuN$ , the result follows by appealing to 
Lemma~\ref{lem:Comega23excludes}. 
\end{proof}

The third lemma of this subsection takes care of the third kind of ternary operations that occurrs in
complexity classifications of CSPs for reducts of  homogeneous graphs.

\begin{lemma}
\label{lem:implexcludbyh}
Let $\mathbb{A}$ be a reduct of a countably infinite homogeneous graph preserved by a ternary 
canonical operation  $h$ with $h(N, \cdot, \cdot) = h(\cdot,N, \cdot) = h(\cdot, \cdot,N) = N$ and which behaves like a
minority on $\{E,=\}$, i.e., $h$ satisfies the behaviour $B$ such that $B(E,E,E) = B(E,=,=) = B(=,E,=) = B(=,=, E) = E$
and $B(=,=,=) = B(=,E,E) = B(E,=,E) = B(E,E,=)$ equals $=$. Then $\mathbb{A}$ pp-defines none of the following types of relations:
\begin{itemize}
\item $[(N(x_1,x_2) \implies \uuE(x_3, x_4))]$-relations,
\item $[(N(x_1,x_2) \implies x_3 = x_4)]$-relations, 
\item $[(N(x_1,x_2) \implies E(x_3, x_4)), (\uuE(x_3, x_4))]$-relations. 
\end{itemize}
\end{lemma}

\subsection{Types of Relations violated by Oligopotent QNUs}
\label{sect:violatedbyQNU}

Here we provide a list of quaternary relations of interest violated by ternary canonical operations and oligopotent qnu-operations. We start with the case where the considered homogeneous graph is the random graph.

\begin{lemma}
\label{lem:atmostoneENorNE}
Let $\mathbb{A}$ be a first-order expansion of the random graph 
preserved by a ternary injection of behaviour majority which additionally satisfies one of the conditions in Lemma~\ref{lem:graphmajorityequalities}  or of behaviour minority which additionally satisfies one of the conditions in
Lemma~\ref{lem:graphminorityequalities}, and an oligopotent 
qnu-operation. Then $\mathbb{A}$ pp-defines at most one of the following:
\begin{enumerate}
\item \label{atmostone:ENequal} either a $[(E(x_1, x_2) \implies \uuN(x_3, x_4))]$-relation or
\item \label{atmostone:NEequal} a $[(N(x_1,x_2) \implies \uuE(x_3,x_4))]$-relation.
\end{enumerate}
\end{lemma}

\noindent
Here comes the corollary.

\begin{corollary}
\label{cor:ENternaryQNU}
 Let $\mathbb{A}$ be a first-order expansion of the random graph 
preserved by a ternary injection of behaviour majority which additionally satisfies one of the conditions in Lemma~\ref{lem:graphmajorityequalities}  or of behaviour minority which additionally satisfies one of the conditions in
Lemma~\ref{lem:graphminorityequalities} and an oligopotent 
qnu-operation. Then $\mathbb{A}$ has relational width $(2,3)$.
\end{corollary}

\begin{proof}
By appeal to Lemma~\ref{lem:atmostoneENorNE}, it follows that there are $\{ O_1, O_2 \} = \{ E, N \}$ 
such that $\mathbb{A}$ does not pp-define a $[(O_1(x_1,x_2) \implies \underline{\underline{O_2}}(x_3, x_4))]$-relation. By Lemmas~\ref{lem:graphmajorityequalities} and~\ref{lem:graphminorityequalities}, $\mathbb{A}$
pp-defines neither $[(O_1(x_1, x_2) \implies x_3 = x_4)]$-relations nor $[(O_2(x_1, x_2) \implies x_3 = x_4)]$-relations.
Since $\maxbound_{\homograph}$ in the case where $\homograph$ is the random graph equals $3$, the result follows by Lemma~\ref{lem:injO1dominating}.
\end{proof}

We now turn to the case where the considered homogeneous graph is the disjoint union of $\omega$ edges.

\begin{lemma}
\label{lem:noEEqEqEImplications}
Let $\mathbb{A}$ be a first-order expansion of $C^2_{\omega}$ 
preserved by a ternary 
canonical operation  $h$ with $h(N, \cdot, \cdot) = h(\cdot,N, \cdot) = h(\cdot, \cdot,N) = N$ and which behaves like a
minority on $\{E,=\}$ and an oligopotent qnu-operation.
Then $\mathbb{A}$  pp-defines neither
\begin{itemize}
\item a $[(E(x_1, x_2) \implies (x_3 = x_4)),(\uuE(x_1, x_2) \wedge N(x_2, x_3) \wedge \uuE(x_3, x_4))]\textrm{-relation}$ nor
\item  a $[((x_1 = x_2) \implies E(x_3, x_4)),(\uuE(x_1, x_2) \wedge N(x_2, x_3) \wedge \uuE(x_3, x_4))]\textrm{-relation}$.
\end{itemize}
\end{lemma}

\noindent
Then we provide another similar lemma.

\begin{lemma}
\label{lem:noEEEqEqImplications}
Let $\mathbb{A}$ be a first-order expansion of $C^2_{\omega}$ 
preserved by a ternary 
canonical operation  $h$ with $h(N, \cdot, \cdot) = h(\cdot,N, \cdot) = h(\cdot, \cdot,N) = N$ and which behaves like a
minority on $\{E,=\}$ and an oligopotent qnu-operation.
Then $\mathbb{A}$  pp-defines neither
\begin{itemize}
\item a $[(E(x_1, x_2) \implies E(x_3, x_4)),(\uuE(x_1, x_2) \wedge N(x_2, x_3) \wedge \uuE(x_3, x_4))]$-relation nor
\item a  $[((x_1 = x_2) \implies (x_3 = x_4)),(\uuE(x_1, x_2) \wedge N(x_2, x_3) \wedge \uuE(x_3, x_4))]$-relation.\end{itemize}
\end{lemma}

\noindent
Then comes the corollary.

\begin{corollary}
\label{cor:EEqternaryQNU}
 Let $\mathbb{A}$ be a first-order expansion of $C^2_{\omega}$ 
preserved by a ternary 
canonical operation  $h$ with $h(N, \cdot, \cdot) = h(\cdot,N, \cdot) = h(\cdot, \cdot,N) = N$ and which behaves like a
minority on $\{E,=\}$ and an oligopotent qnu-operation.
Then $\mathbb{A}$ has relational width $(2,3)$.
\end{corollary}

\begin{proof}
By Lemmas~\ref{lem:implexcludbyh},~\ref{lem:noEEqEqEImplications} and~\ref{lem:noEEEqEqImplications} none of the types of relations mentioned in Lemma~\ref{lem:hnear23min} is pp-definable in $\mathbb{A}$. Appealing to
Lemma~\ref{lem:hnear23min} completes the proof of the corollary. 
\end{proof}

\subsection{The Main Result}
\label{sect:mainresult}

Here we prove our main result.

\begin{theorem}
\label{thm:mainresult}
Let $\mathbb{A}$ be a first-order expansion of a countably infinite homogeneous  graph $\homograph$ which has bounded strict width. Then $\mathbb{A}$ has relational width $(2,\maxbound_{\homograph})$.
\end{theorem}

\begin{proof}
By the classification of Lachlan and Woodrow~\cite{LachlanWoodrow}, we have that $\homograph$ is either the random graph, a Henson graph $H_k$ with a forbidden $k$-clique where $k \geq 3$, 
a disjoint set of $n$ cliques of size $s$ denoted by $C^s_n$ or a complement of either $C^s_n$ or $H_k$. 
The case where $\homograph$ is $C^{\omega}_1$, $C^{1}_{\omega}$, $C^{\omega}_{\omega}$ or $H_k$ with $k \geq 3$ follows by Corollary~\ref{cor:newDich}.
If $C^s_n$ is such that $3 \leq n < \omega$ or 
$3 \leq s < \omega$, then by Theorem~60 in~\cite{homographs-arxiv}, a first-order expansion $\mathbb{A}$ 
of $C^s_n$ is either homomorphically equivalent to a reduct of $(\N; =)$ or is not preserved by an oligopotent qnu-operation and we are done.
If $\mathbb{A}$ is a first-order expansion of $C^{\omega}_2$, then by Theorem~61 in~\cite{homographs-arxiv} either it is homomorphically equivalent to a reduct of $(\N;=)$ or
is not preserved by an oligopotent qnu-operation or pp-defines both $E$ and $N$ and   is preserved by a canonical ternary injection of behaviour minority which is hyperplanely balanced of behaviour xnor and then $\mathbb{A}$ has relational width $(2,3)$ by Corollary~\ref{cor:Comega23min}. If $\mathbb{A}$ is a first-order expansion of $C^2_{\omega}$, then by Theorem~62
in~\cite{homographs-arxiv}, we have that either $\mathbb{A}$ is homomorphically equivalent to a reduct of $(\N; =)$, or it is not preserved by an oligopotent qnu-operation or it pp-defines both $E$ and $N$ and is preserved by a canonical binary injection  of behaviour min that is $N$-dominated or a ternary canonical operation $h$ with $h(N,\cdot,\cdot) = h(\cdot,N,\cdot) = h(\cdot,\cdot,N) = N$ and which behaves like a minority on $\{E,=\}$. In the former case the language $\mathbb{A}$ has relational width $(2,3)$ by Lemma~\ref{lem:graphreductbinary}, in the latter by Corollary~\ref{cor:EEqternaryQNU}.

The remaining case is where $\structA$ 
is a first-order expansion of  the random graph $G$ preserved by an oligopotent qnu-operation. By Theorem 6.1 in~\cite{BodirskyP15} we have that a first-order expansion of $G$ is either homorphically equivalent to a reduct of $(\N; =)$
and then we are done or pp-defines both $E$ and $N$ in which case, by 
Theorem~9.3 in~\cite{BodirskyP15}, we have that:
\begin{itemize}
\item $\mathbb{A}$ is preserved by  a binary injection  of behaviour max which is either balanced or E-dominated, 
by a binary injection  of behaviour min which is either balanced or $N$-dominated, 
by a binary injection which is E-constant, or a binary injection  which is  $N$-constant,
and then the theorem follows by Lemma~\ref{lem:graphreductbinary}, or 
\item $\mathbb{A}$ is preserved by a ternary injection of behaviour majority which additionally satisfies one of the conditions in Lemma~\ref{lem:graphmajorityequalities}  or of behaviour minority which additionally satisfies one of the conditions in
Lemma~\ref{lem:graphminorityequalities}, and then the theorem holds by Corollary~\ref{cor:ENternaryQNU}.
\end{itemize}
It completes the proof of the theorem.
\end{proof}

\section{Summary and Future Work}

In this paper we proved in particular that:
\begin{enumerate}
\item  \label{sum:binaryop} every first-order expansion 
of a homogeneous graph $\homograph$ preserved by a canonical binary operation considered in~\cite{BodirskyP15,equiv-csps,homographs-arxiv} and
\item \label{sum:stric} every first-order expansion of a homogeneous graph $\homograph$ with bounded strict width has relational width exactly $(2, \maxbound_{\homograph})$.
\end{enumerate}
A nice consequence of the former result is that all tractable reducts of $C^1_{\omega}, C^{\omega}_1, C^{\omega}_{\omega}$ and $H_{k}$ with $k \geq 3$ have relational width exactly $(2, \maxbound_{\homograph})$, and thus all tractable $\csp(\mathbb{A})$ may be solved by establishing $(2, \maxbound_{\homograph})$-minimality. Nevertheless, we find the latter result to be the main result of this paper. It is for the following reason.

Our general strategy is that we show that constraint languages $\mathbb{A}$ under consideration do not express `too many implications', i.e., quaternary relations that efficiently entail formulas of the form $(R_1(x_1, x_2) \implies R_2(x_3, x_4))$, see definitions in Section~\ref{sect:effent} and lemmas in Section~\ref{sect:suffconditions} and then use these facts in order to find a strategy of how to shrink a non-trivial $(2, \maxbound_{\homograph})$-minimal instance of the CSP so that it became a simple instance. 
In this paper, in order to show that certain relations are not pp-definable in $\structA$ we employ in particular some binary and ternary canonical operations. We believe that it is not in fact necessary and theorems analogous to Theorem~\ref{thm:mainresult} may be obtained for large families of constraint languages  using only the fact that structures $\mathbb{A}$ under consideration are preserved by oligopotent  qnu-operations. 
Thereby we believe that Question~\ref{quest:relstrictwidth} may be answered in full generality.

\bibliography{STACS.bib}

\appendix

\section{Proof of Proposition~\ref{prop:minimality}}

 The algorithm of establishing $(k,l)$-minimality is quite straightforward and adopted from the world of finite CSP.
 First for every $l$-element subset $\{ v_{i_1}, \ldots, v_{i_l} \}$ of the variables $\V = \{ v_1, \ldots, v_n \}$ in $\varphi_I$ that does not occurr in any constraint in $\varphi_I$ we introduce
 a constraint $(( v_{i_1}, \ldots, v_{i_l}), A^l)$.
  Then until the process stabilizes, for every $W \subseteq \V$ of size at most $k$ and every two constraints 
 $((v_{i_1}, \ldots, v_{i_l} ), R_1)$, $((v_{j_1}, \ldots, v_{j_k}), R_2)$  such that $W$  is a subset of both $\{ x_1, \ldots, x_a \}$ and
 $\{ y_1, \ldots, y_b \}$ we replace $R_1$ with 
 $R'_1(v_{i_1}, \ldots, v_{i_l} ) \equiv 
\exists z_1 \cdots \exists z_c~R_1(v_{i_1}, \ldots, v_{i_l} ) \wedge R_2(v_{j_1}, \ldots, v_{j_k})$
 where $\{ z_1, \ldots, z_c \} = \{ v_{i_1}, \ldots, v_{i_l} \} \setminus  \{ v_{j_1}, \ldots, v_{j_k} \}$; and we replace $R_2$ with $R'_2(v_{j_1}, \ldots, v_{j_k}) \equiv 
 \exists z_1 \cdots \exists z_c~R_1(v_{i_1}, \ldots, v_{i_l} ) \wedge R_2(v_{j_1}, \ldots, v_{j_k})$
 where $\{ z_1, \ldots, z_c \} = \{ v_{i_1}, \ldots, v_{i_l} \} \setminus  \{ v_{j_1}, \ldots, v_{j_k} \}$.
 
 We introduce $O(\left| \V \right|^l)$ new constraints and then we do the replacement  at most 
 the number of constraints multiplied by the number of pairwise different orbits of tuples of the length 
 of the longest constraint.  Indeed, in every step of the algorithm at least one orbit from at least one constraint is removed.
 Since there are at most $O(\left| \V \right|^l)$ constraints and the length of every constraint is bounded by a constant: $\max(l, m)$
 where $m$ is the largest arity in the signature of $\mathbb{A}$. It follows that for fixed $(k,l)$ the process of establishing $(k,l)$-minimality may be performed in time $O(\left| \V \right|^l)$. $\square$

\section{Proof of Corollary~\ref{cor:relwidth}}

Here we use notions and results from~\cite{UnaryDichotomy}. 
Let $(b_1, \ldots, b_m)$ be a tuple of elements in $\mathbb{B}$. A quantifier-free (qf-) type of  $(b_1, \ldots , b_m)$, also called an $m$-type in B, is the set of all quantifier-free formulas $\phi(z_1, \ldots, z_m)$ such that $\mathbb{B} \models \phi(b_1,\ldots , b_m)$. The structures $\mathbb{B}$ under considerations have finite relational signatures, and hence there are only finitely many $m$-types in $\mathbb{B}$. The set of all $m$-types in $\mathbb{B}$ will be denoted by
$T_{\mathbb{B}, m}$.

Define now $T_{\mathbb{B}, m}(\mathbb{A})$ to be a finite structure whose elements are the $m$-types $T_{\mathbb{B}, m}$ of 
$\mathbb{B}$ and relations are as follows. 
\begin{itemize}
\item We have a unary constraint for each relation $R$ of $\mathbb{A}$ of arity $r$ with a definition $\chi(z_1,\ldots,z_r)$ in $\mathbb{B}$. For $i: [r] \to [m]$ we write $\langle \chi(z_{i(1)},\ldots, z_{i(r)} \rangle$ to denote 
the unary relation that consists of all the types that contain $\chi(z_{i(1)}, \ldots  , z_{i(r)})$, and add all such relations to $T_{\mathbb{B},m(\mathbb{A})}$\footnote{As in~\cite{UnaryDichotomy}, we use functions to index tuples. }.
\item We also have a binary relation 
$\textrm{Comp}_{i,j}$
for all $r \in [m]$ and $i, j : [r] \to [m]$. Define $\textrm{Comp}_{i,j}$ to be the binary relation that contains the pairs $(p, q)$ of $m$-types such that for every quantifier-free formula $\chi(z_1, \ldots , z_s)$ of $\mathbb{B}$ and 
$t: [s] \to [r]$, the formula $\chi(z_{it(1)},\ldots,z_{it(s)})$ is in $p$ iff $\chi(z_{jt(1)},\ldots,z_{jt(s)})$ is in $q$.
\end{itemize}

In the proof of Theorem 3.1 in~\cite{UnaryDichotomy}, an instance $\mathcal{I}_A$ 
of $\csp(\mathbb{A})$ over variables $\V = \{v_1,\ldots,v_n\}$
is transformed into an instance $\mathcal{I}_T$ of $\csp(T_{\mathbb{B}, m}(\mathbb{A}))$  
over variables that are increasing functions $I$ from $[m]$ to $\V$. The intention is  that $h: \V \rightarrow A$ is a solution to $\mathcal{I}_A$  
if and only if  $g: I \to T_{\mathbb{B}, m}$ such that $g(v)$ is the type of 
$(h(v(1), \ldots, h(v(m))$ is a solution to $\mathcal{I}_T$.  
\begin{itemize}
\item The instance $\mathcal{I}_T$ has
a unary constraint for every  constraint  $(((j(1), \ldots, j(r)), R)$
in $\mathcal{I}_A$ where $R$ is a relation of $\mathbb{A}$
with a qf-definition $\chi(z_1,\ldots,z_r)$ over $\mathbb{B}$
 and $j:[r] \to V$.  Let $v \in I$ be such that $\textrm{Im}(j) \subseteq  \textrm{Im}(v)$ and $U$ be the relation symbol of $T_{\mathbb{B},m}(\mathbb{A})$ that denotes the unary relation $\langle \chi(z_{v^{-1}j(1)}, \ldots , z_{v^{-1}j(r)}) \rangle$. We then add $((v),U)$ to $\mathcal{I}_T$.
\item Furthermore, $\mathcal{I}_T$ has a binary constraint 
$((v,s), \textrm{Comp}_{v^{-1}k,s^{-1}k})$ for all $v,s \in I$ and bijections 
$k: [r] \to \textrm{Im}(v) \cap \textrm{Im}(s)$. 
\end{itemize}

\noindent
By the proof of Theorem 3.1 in~\cite{UnaryDichotomy} we have the following.
\begin{fact}
\label{fact:solutionsTransfer}
The assignment $h: \V \rightarrow A$ is a solution to $\mathcal{I}_A$  
if and only if  $g: I \to T_{\mathbb{B}, m}$ such that $g(v)$ is the type of 
$(h(v(1), \ldots, h(v(m))$ is a solution to $\mathcal{I}_T$.  
\end{fact}

By Lemma 4.7 in~\cite{UnaryDichotomy}, the structure $T_{\mathbb{B}, m}(\mathbb{A})$ has a ternary and a four-ary polymorphism required 
for bounded width. It folows that $T_{\mathbb{B}, m}(\mathbb{A})$ has bounded width $(2,3)$, and hence $\csp(T_{\mathbb{B}, m}(\mathbb{A}))$ is solvable by establishing $(2,3)$-minimality. For the proof of the corollary we will show that  
the algorithm of establishing minimality transforms an instance 
$\mathcal{I}_A$ of $\csp(\mathbb{A})$ to a non-trivial $(2m,3m)$-minimal
instance  iff the algorithm of establishing minimality transforms an instance 
$\mathcal{I}_T$ of $\csp(T_{\mathbb{B}, m}(\mathbb{A}))$  to a $(2,3)$-minimal instance.

The right-to-left implication is easy to observe. Indeed,  if  an instance $\mathcal{I}_T$ of $\csp(T_{\mathbb{B}, m}(\mathbb{A}))$  has  a $(2,3)$-minimal non-trivial equivalent instance of the CSP, then $\mathcal{I}_T$ has a solution.
By Fact~\ref{fact:solutionsTransfer}, we have that $\mathcal{I}_A$ has also a solution and hence there is a $(2m,3m)$-minimal non-trivial instance of the CSP equivalent to $\mathcal{I}_A$.

For the left-to-right implication consider a non-trivial $(2m, 3m)$-minimal instance $\mathcal{I}'_A$ of the CSP equivalent to $\mathcal{I}_A$
that is obtained from $\mathcal{I}_A$ by first adding an extra constraint 
$((j(1), \ldots j(r)), A^r)$ for all $m \leq r \leq 3m$ and all $j: [r] \to [m]$ and then establishing $(2m, 3m)$-minimality. In consequence we obtain a $(2m, 3m)$-minimal instance with the following constraints:
\begin{itemize}
\item $((j(1),\ldots, j(r)), R')$ for every constraint of the form $((j(1),\ldots, j(r)), R)$ in $\mathcal{I}_A$ where $R' \subseteq R$, and 
\item a constraint $((j(1),\ldots, j(r)), S')$ for all $m \leq r \leq 3m$ and $j:[r] \rightarrow \V$ originating from a constraint $((j(1), \ldots j(r)), A^r)$.
\end{itemize}

The next step is to construct  $\mathcal{I}'_T$ from
$\mathcal{I}'_A$ and $\mathcal{I}_T$. To this end, for $u_1, \ldots, u_k \in \textrm{I}$ and a constraint $C = (\overline{x}, R)$ whose scope contain $\textrm{Im}(u_1), \ldots, \textrm{Im}(u_k)$ we define 
$\mathrm{Types}(C,u_1, \ldots, u_k)$ to be a set of
$k$-tuples of $m$-types which for every tuple $t \in R$ 
contains $(\textrm{type}_1, \ldots, \textrm{type}_k)$ where 
$\textrm{type}_i$ is the type of the projection of  $t$ to $\textrm{Im}(u_i)$.
We are now ready to define
$\mathcal{I}'_T$.

\begin{enumerate}
\item \label{ITprim:unarynew} For a constraint $C = ((j(1), \ldots, j(m)), S')$ in $\mathcal{I}'_A$ we have 
$(j, \textrm{Types}(C, j))$ in $\mathcal{I}'_T$.
\item \label{ITprim:binarynew} For a constraint $C = ((j(1), \ldots, j(r)),S')$ with $m < r \leq 2m$
and all  $u,v \in I$ with $\textrm{Im}(u), \textrm{Im}(v) \subseteq \textrm{Im}(j)$ we have a binary constraint $((u,v), \textrm{Types}( C, u, v))$ in $\mathcal{I}'_T$.
\item \label{ITprim:ternarynew} For a constraint $((j(1), \ldots, j(r)),S')$ with $m+2 < r \leq 3m$
and all  $u,v, w \in I$ with $\textrm{Im}(u), \textrm{Im}(v), \textrm{Im}(w)  \subseteq \textrm{Im}(j)$ we have a ternary constraint $((u,v, w), (\textrm{Types}(C,u,v,w))$ in $\mathcal{I}'_T$. 
\item \label{ITprim:unaryold} Let $((v), U)$ be a unary constraint in $\mathcal{I}_{T}$
and $(j, \textrm{Types}(C, j))$ a constraint from Item~\ref{ITprim:unarynew} above. Then 
$\mathcal{I}'_{T}$ contains $((v), U \cap \textrm{Types}(C, j))$.

\item \label{ITprim:binaryold} Let $((v,s), \textrm{Comp}_{v^{-1}k,s^{-1}k})$ be a binary constraint in $\mathcal{I}_{T}$ and $((v,s), \textrm{Types}(C, v,s))$ a constraint from Item~\ref{ITprim:binarynew}. Then $\mathcal{I}'_{T}$ contains $((v,s), \textrm{Comp}_{v^{-1}k,s^{-1}k}  \cap  \textrm{Types}(C, v,s))$.
\end{enumerate}

In order to complete the proof of the corollary, we need to show that 
 $\mathcal{I'}_T$ is  non-trivial, $(2,3)$-minimal and equivalent to $\mathcal{I_T}$.
 
\begin{fact}
The instance $\mathcal{I'}_T$ is equivalent to $\mathcal{I}_{T}$.
\end{fact}

\begin{proof}
By Fact~\ref{fact:solutionsTransfer} we have that
$h: \V \rightarrow A$ is a solution to $\mathcal{I}_A$  
if and only if  $g: I \to T_{\mathbb{B}, m}$ such that $g(v)$ is the type of 
$(h(v(1), \ldots, h(v(m))$ is a solution to $\mathcal{I}_T$.  We need to show that every such $g$ is also a solution to $\mathcal{I}'_T$. Since $h$ satisfies all new constraints in $\mathcal{I}_A$, i.e., constraints of the form $((j(1),\ldots, j(r)), S')$ for all $m \leq r \leq 3m$ and $j:[r] \rightarrow \V$ originating from a constraint $((j(1), \ldots j(r)), A^r)$, it is straightforward to check that $g$ satisfies all constraints in $\mathcal{I}'_T$ from Items~\ref{ITprim:unarynew}--\ref{ITprim:ternarynew}.
 Since $g$ satisfies these new constraints  and
all constraints in $\mathcal{I}_T$, it also satisfies constraints in $\mathcal{I}'_T$
defined in Items~\ref{ITprim:unaryold} and~\ref{ITprim:binaryold}.
\end{proof}

\begin{fact}
The instance $\mathcal{I}'_T$ is non-trivial.
\end{fact}
 
\begin{proof}
We need to show that all relations associated to  constraints in $\mathcal{I}'_T$
are non-empty. Since $\mathcal{I}'_A$ is non-trivial, this statement clearly holds for all constraints from Items~\ref{ITprim:unarynew}--\ref{ITprim:ternarynew}.
Since all constraints in $\mathcal{I}'_T$ from Item~\ref{ITprim:unarynew} are non-empty, we have that for every $v \in I$
there is an $m$-type which satisfies all unary constraints in $\mathcal{I}_T$
of the form $((v), U)$. It follows that relations in constraints in $\mathcal{I}'_T$ in Item~\ref{ITprim:unaryold} are non-empty. Finally observe that for every 
constraint $((v,s), R)$ in $\mathcal{I}'_T$ from Item~\ref{ITprim:binarynew},
the 
binary relation $R$  is contained in all relations
of the form $\textrm{Comp}_{v^{-1}k,s^{-1}k}$. It follows that relations in constraints defined in Item~\ref{ITprim:binaryold} are also non-empty. It completes the proof of the fact.
\end{proof} 
 
\begin{fact}
The instance $\mathcal{I'}_T$ is $(2,3)$-minimal.
\end{fact}

\begin{proof}
Clearly all $s,u,v \in I$ occur together in the scope of some constraint in $\mathcal{I}'_T$ from Item~\ref{ITprim:ternarynew}. 
For Condition~\ref{M2} in Definition~\ref{def:minimality},
observe both of the following:
\begin{itemize}
\item for all $v \in I$, the constraint $(v,U)$ in $\mathcal{I}'_T$
in Item~\ref{ITprim:unaryold} and all constraints $(v, U')$ in Item~\ref{ITprim:unarynew}, we have that $U = U'$, and
\item for all $ v,u \in I$, the constraint $((v,u), R)$ in $\mathcal{I}'_T$
in Item~\ref{ITprim:binaryold} and all constraints $((v,u), R')$ in Item~\ref{ITprim:binarynew}, we have that $R = R'$.
\end{itemize}
Thus, in order to satisfy Condition~\ref{M2}, it is enough to restrict to the case
where $C_1, C_2$ are constraints from Items~\ref{ITprim:unarynew}--\ref{ITprim:ternarynew}. But
Condition~\ref{M2} in Definition~\ref{def:minimality} holds for all pairs of constraints in Items~\ref{ITprim:unarynew}--\ref{ITprim:ternarynew} because
 $\mathcal{I}'_{A}$ is $(2m, 3m)$-minimal. 
It completes the proof of the fact.
\end{proof}

\section{Proof of Proposition~\ref{prop:binRW}}

Let $h$ be any binary canonical operation from the formulation of the lemma.
We will show that in any case: 
\begin{itemize}
\item $f(x_1, x_2, x_3, x_4) = h(h(h(x_1, x_2), x_3), x_4)$ and
\item $g(x_1, x_2, x_3) = h(h(x_1, x_2), x_3)$
\end{itemize}
are weak near-unanimity operations modulo $\overline{\Aut(\homograph)}$
where $\homograph$ is a homogeneous graph and that they satisfy the 
condition in Corollary~\ref{cor:relwidth}. Clearly, $f$ and $g$ are canonical.
In order to complete the proof of the proposition, we will show that all :
\begin{itemize}
\item $f(y,x, x,x), f(x,y,x,x), f(x,x,y,x)$ and $f(x,x,x,y)$, 
\item $g(y,x,x), g(x,y,x), g(x,x,y)$, and 
\item $f(y,x,x,x)$ and $g(y,x,x)$ have the same behaviour.
\end{itemize}

Observe that if a behaviour $B$ of $h$  is min or max, then 
it is a similattice operation on the three-element set $\{ E, N, = \}$.
It is therefore straightforward to check that all operations obtained from $f$ and $g$ mentioned in the items above are of the same behaviour as $h$. Thus, the proposition follows in this case where $h$ is of behaviour min or max.

If $h$ is $E$-constant or $N$-constant, then  
all operations mentioned in the items above are $E$-constant or $N$-constant, respectively. Thus, the proposition follows also in this case.
$\square$

\section{Proof of Observation~\ref{obs:minequality}}

The relation $\sim$  is clearly reflexive and symmetric. Assume on the contrary that it is not 
transitive. Then there are $i,j,k \in [n]$ such that $\instance_{i,j}, \instance_{j,k} \subseteq \{ = \}$ and $\instance_{i,k}$ contains 
$O \in \{E, N \}$. Since $\mathcal{I}$ is $(2, \maxbound)$-minimal and $\maxbound \geq 3$, it contains a constraint $C$ whose scope contains $\{v_i, v_j, v_k\}$. But $C$ is not satisfiable by $a: \{ v_i, v_j, v_k \} \rightarrow D$ 
such that $a(v_i) = a(v_j), a(v_j) = a(v_k)$ and $(a(v_i), a(v_k)) \in O$. It contradicts $(2, \maxbound)$-minimality of 
$\mathcal{I}$ and completes the proof of the observation.
$\square$

\section{Proof of Proposition~\ref{prop:ClausesQNU}}

In order to prove the proposition, we will show that any $\mathbb{A}$ is preserved by an oligopotent qnu-operation, see Theorem~\ref{thm:strictwidth}.
We start with a couple of definitions.

A $k$-tuple $t$ with $k \geq 3$ has the main value $a$ if all but at most one entry in $t$ takes the value $a$. Note that not all tuples have the main value, but if a tuple has the main value, then it is unique. It is straightforward to see that every qnu-operation sends all tuples with the same main value to the same element of the image. 
For a natural number $k \geq 3$, define $\sim_k$ to be the equivalence relation on $A^k$ such that $u,v \in A^k$ satisfy $u \sim_k v$ if and only if $u$ and $v$ have the same main value. We first state a simple observation.
   
\begin{observation}
\label{obs:appropQNU}
Let  $p =2m+1$ for some natural number $m$  and $R \in \{ E, N\}$.
Then  there exists an oligopotent quasi near-unanimity operation $f$ of arity $p$ 
such that all of the following hold:
\begin{enumerate}
\item \label{qnu:equality} $f(u) = f(v)$ if and only if $u \sim_{p} v$;
\item \label{qnu:preservesEN} $f$ preserves both $E$ and $N$;
\item \label{qnu:twomainvalues} if $u$ has the main value $c$ and $v$ has the main value $d$, then $(f(u), f(v)) \in R$ if and only if $(c, d) \in R$;
\item \label{qnu:onemainvalue} if $u$ has the main value $c$ and $v$ has no main value and 
there exists $I \subseteq [p]$ with $\left| I \right| \geq m+1$ such that $(c, v[i]) \in R$ for all 
$i \in I$
then $(f(u),f(v)) \in R$.
\end{enumerate}
\end{observation}

\begin{proof} Consider a graph $G_f$
whose vertices are the equivalence classes  of $\sim_p$ and $[u]_{\sim_p}$ and $[v]_{\sim_p}$ in $A^{p}/ \sim_p$ 
are connected by an edge if one of the following holds:
\begin{itemize}
\item $u$ has the main value $c$,  $v$ has the main value $d$  and $(c,d) \in E$,
\item $u$ has the main value $c$,  $v$ has no main value and $I = \{ i \in [2m+1] \mid  (c, v_i) \in E \}$ satisfies $\left| I \right| \geq m+1$,
\item neither $u$ nor $v$ has the main value and $(u[i], v[i]) \in E$ for all $i \in [2m+1]$.
\end{itemize}
Since $G_f$  is a countable graph, it clearly embeds into the random graph by an embedding $f$. It is now straightforward to check that $f$ satisfies all 4 items in the formulation of the observation.
\end{proof}

For every relation $R$ in $\mathbb{A}$ set $\varphi_R$ to be some definition of $R$ given by the conjunction of clauses of the form given in the formulation of the 
proposition and $r = k + l$ be the maximal number of disequalities over all clauses in all $\varphi_R$ for all $R$ in $\mathbb{A}$. Let $m$ be a natural number satisfying
$m \geq r + 2$.  We will prove that 
$\mathbb{A}$ is preserved by a $(2m+1)$-ary oligopotent qnu-operation
 $f$ from Observation~\ref{obs:appropQNU}. By appeal to Theorem~\ref{thm:strictwidth}, it will imply that $\mathbb{A}$ has bounded strict width.

In fact, by Theorem~\ref{thm:Galoisconn}, it is enough to prove that $f$ 
preserves any relation defined by a single clause of the desired form with the number of disequalities $(k + l)$ less than $r$. 
Assume on the contrary that $f$ violates some $R$ defined by a single clause $\psi$ 
of the considered form over variables $X = \{ x_1, \ldots, x_k, y_1, y_2, z_1, \ldots, z_l \}$
 and 
set $a(v)$ to be the vector $(a_1(v), \ldots, a_{2m+1}(v))$ for all $v \in X$ 
where $a_1, \ldots, a_{2m+1}: X \to A$ 
are such that all $a_i$ with $i \in [2m+1]$ satisfy $\psi$ but $h: X \to A$
defined so that $h(v) = f(a_1(v), \ldots, a_{2m+1}(v))$ for all $v \in X$ does not satisfy $\psi$.
In particular, $h(y_1) = h(x_i)$ for all $i \in [k]$, $h(y_2) = h(z_i)$ for all $i \in [l]$
 and $(h(y_1), h(y_2)) \notin R$. Assume without loss of generality that 
all $a(x_1), \ldots, a(x_k)$ are different than $a(y_1)$ and that 
all $a(z_1), \ldots, a(z_l)$ are different than $a(y_2)$.
Then we have both of the following:
\begin{itemize}
\item either $\{ x_1, \ldots, x_k \}$ is empty  or $a(y_1), a(x_1), \ldots, a(x_k)$ have all the same main value $c$, and
\item either $\{ z_1, \ldots, z_l \}$ is empty  or  
$a(y_2), a(z_1), \ldots, a(z_l)$
have the same main value $d$. 
\end{itemize}

If both $\{ x_1, \ldots, x_k \}$ and $\{ z_1, \ldots, z_l \}$ are empty, then $\psi$ 
is $R(y_1, y_2)$ and we have the contradiction.
Indeed, by Item~\ref{qnu:preservesEN} in Observation~\ref{obs:appropQNU}, $f$ preserves both $E$ and $N$. 
Next, turn to the case where exactly one holds:
either $\{ x_1, \ldots, x_k \}$ or $\{ z_1, \ldots, z_l \}$ is empty.
Assume without loss of generality that $\{ z_1, \ldots, z_l \}$ is empty.
We now consider a couple of 
cases. Observe that in this case   $a(y_1)$ has the main value $c$. If $a(y_2)$ also has  the main value, then for all but at most two $i \in [2m+1]$ we have $(a_i(y_1)), a_i(y_2)) \in R$. Since $m \geq 2$, by the definition of $f$, it follows that $ 
(h(y_1), h(y_2)) \in R$.
Thus, we may assume that $a(y_2)$
does not have the main value. 
Set $I \subseteq [2m+1]$ to be the set of indices such that $i \in I$
if and only if $(a_i(x_j) = a_i(y_1))$ for all $j \in [k]$. 
Observe that $(a_i(y_1), a_i(y_2)) \in R$ for all $i \in I$.
If $(a_1(y_1), \ldots, a_{2m+1}(y_1))$ is a constant tuple, then $\left| I \right| \geq m+1$, then, by Item~\ref{qnu:onemainvalue} in Observation~\ref{obs:appropQNU}, we have $(h(y_1), h(y_2)) \in R$. 
If $a(y_1)$ is not a constant tuple, then either $i$ such that $a_i(y_1) \neq c$
is in $I$ or not. 
Observe that in the first case $I = [2m+1]$ and that it contradicts the assumption.  In the other case, since $m \geq k+1$, we have that
 $\left| I \right| \geq m+1$ and for all $i \in I$ it holds
$a_i(y_1) = c$. Clearly, for all $i \in I$ we have also that $(a_i(y_1), a_i(y_2)) \in R$.
This time, by Item~\ref{qnu:onemainvalue} in
Observation~\ref{obs:appropQNU}, we have that $(h(y_1), h(y_2)) \in R$. It completes the proof of the case where either $\{ x_1, \ldots, x_k \}$ or $\{ z_1, \ldots, z_l \}$ is empty.

Finally, we turn to the case where neither $\{ x_1, \ldots, x_k \}$ nor $\{ z_1, \ldots, z_l \}$ is empty.
Then both $a(y_1)$ and $a(y_2)$ have the main values $c$ and $d$, respectively.
Since $(c,d) \in R$ we have by Item~\ref{qnu:twomainvalues} in Observation~\ref{obs:appropQNU} that $(h(y_1), h(y_2)) \in R$.  It completes the proof of the proposition.

\section{Proof of Proposition~\ref{prop:qnuExTwoEquivClass}}

In order to prove that $\mathbb{A} = (A; E, N, R)$ where 
$(A; E)$ is $C^{\omega}_2$ and $R(x_1, x_2, x_3) \equiv ((E(x_1, x_2) \wedge N(x_2, x_3)) \vee (N(x_1, x_2) \wedge E(x_2, x_3))$ has bounded strict width,
we will show that it is preserved by an oligopotent qnu-operation. 
To this end, we will see elements of $A$ as pairs in $\{ 0,1\} \times \mathbb{N}$
so that $((d_1, e_1), (d_2, e_2)) \in E$ if and only if $d_1 = d_2$. Let now $m: \{ 0,1 \}^3 \rightarrow \{ 0,1 \}$ be the majority operation,  i.e., it satisfies 
$m(0,0,0) = m(1,0,0) = m(0,1,0) = m(0,0,1) = 0$ and
$m(1,1,1) = m(0,1,1) = m(1,0,1) = m(0,1,1) = 1$ and $i: \mathbb{N}^3 \to \mathbb{N}$, where $\mathbb{N}$ is the set of natural numbers, be any injection. We claim that $f : A^3 \to A$ defined so that it satisfies:
$f((d_1, e_1), (d_2, e_2), (d_2, e_2)) = (m(d_1,d_2, d_3), i(e_1, e_2, e_3))$
is an oligopotent qnu-operation.

The operation is oligopotent since it clearly preserves $E$ and $N$.
In order to see that $R$ is preserved by $f$ consider 
elements $(e_i^j, d_i^j)$ in $A$ with $i, j \in [3]$ such that for all 
$j \in [3]$ we have $((d_1^j, e_1^j),(d_2^j, e_2^j),(d_3^j, e_3^j)) \in R$.
We have to show that 
$(f((d_1^1, e_1^1),(d_1^2, e_1^2),(d_1^3, e_1^3)),
f((d_2^1, e_2^1),$
$(d_2^2, e_2^2),(d_2^3, e_2^3)),
f((d_3^1, e_3^1),(d_3^2, e_3^2),(d_3^3, e_3^3)))$ is in $R$.
Observe that two out of $\{  d_1^2, d_2^2, d_3^2 \}$ need to have the same value.
Assume without loss of generality that $d_1^2 =  d_2^2 = 0$.
Further, if $d_1^1 = d_2^1 = 0$ and $d_1^3 = d_2^3 = 1$ or 
$d_1^1 = d_2^1 = 1$ and $d_1^3 = d_2^3 = 0$, then we are clearly 
done.
Hence, without loss of generality we may assume that 
$d_1^1  = 0, d_2^1 = 1$ and  $d_1^3 = 1$, $ d_2^3 = 0$. But in this case
exactly one out of $\{ d_1^3, d_3^3 \}$ equals $0$ and still $f$ preserves $R$.
$\square$

\section{Proof of Lemma~\ref{lem:injO1dominating}}
Let $\mathcal{I}$ be a non-trivial $(2, \maxbound)$-minimal instance of the CSP equivalent to an instance of $\csp(\mathbb{A})$ and let 
$\varphi_I$ be a conjunction of atomic formulas containing $R(v_{i_1}, \ldots, v_{i_r})$ for every constraint of the form $((v_{i_1}, \ldots, v_{i_r}), R)$ in $\mathcal{I}$~\footnote{It is well-known that all such relations $R$ are pp-definable in $\mathbb{A}$}.
Then consider 
\begin{align}
\varphi_S &:= \varphi_S^{O_1} \wedge \varphi_S^{O_2} \wedge \varphi_S^{=} \nonumber
\end{align}
where 
\begin{itemize}
\item $\varphi_S^{O_1} := \bigwedge_{\instance_{i,j} \ni O_1} O_1(v_i, v_j)$,
\item $\varphi_S^{O_2} := \bigwedge_{\instance_{i,j} = \{ O_2, = \}} O_2(v_i, v_j) \wedge \bigwedge_{\instance_{i,j} = \{ O_2 \}} O_2(v_i, v_j)$, and
\item $\varphi_S^{=} := \bigwedge_{\instance_{i,j} = \{ = \}} v_i = v_j$.
\end{itemize}

In order to prove that $\mathcal{I}$ has a solution, we will show that  $\varphi_I \wedge \varphi_S$
is satisfiable.
Let $\varphi_I^{O_1}$ be $\varphi_I$ where every atomic formula $\psi$ is replaced by 
$\psi^{O_1}$ equivalent to 
$\psi \wedge \eta_S^{O_1}$ where $\eta_S^{O_1}$ is $\varphi_S^{O_1}$ 
restricted to conjuncts involving variables in $\psi$. We first prove that 
every $\psi \wedge \eta_S^{O_1}$ defines a non-empty relation. Assume the contrary. Then there exists $\psi$, a minimal non-empty subset $\mathbb{I}^{O_1}$
of $\{ (i,j) \mid \instance_{i,j} \ni O_1 \}$ and $(i_3, i_4)$ with $P_{i_3, i_4} \ni O_1$ not in $\mathbb{I}^{O_1}$ 
 such that 
$(\psi \wedge \bigwedge_{(i,j) \in \mathbb{I}^{O_1}} O_1(v_i,v_j))$ entails $\uuOtwo(v_{i_3}, v_{i_4})$.
Let $(i_1, i_2) \in \mathbb{I}^{O_1}$ and $\{ y_1, \ldots, y_l \}$ be all variables occurring in $\psi$.
Set
\begin{align}
\label{eq:quaternary01impO2}
R(x_1, x_2, x_3, x_4) &\equiv (\exists y_1 \cdots \exists y_l~(\psi \wedge \bigwedge_{(i,j) \in \mathbb{I}^{O_1} \setminus \{ (i_1, i_2) \}} O_1(v_i, v_j) \wedge \bigwedge_{j \in [4]} x_j = v_{i_j})).
\end{align}
Since $\mathbb{A}$ is a first-order expansion of $\homograph$ and in consequence contains both $E$ and $N$, we have that $\mathbb{A}$ pp-defines $R$.
Furthermore, observe that $R$ efficiently entails $(O_1(x_1, x_2) \implies \uuOtwo(x_3, x_4))$. The relation $R$ clearly entails  $(O_1(x_1, x_2) \implies \uuOtwo(x_3, x_4))$.
Since $\mathbb{I}^{O_1}$ is minimal,
we have that $R$ contains a tuple $t_1$ such that $(t_1[1], t_1[2]) \in O_1$ and $(t_1[3], t_1[4]) \in \uuOtwo$.
On the other hand, the minimality of $\mathbb{I}^{O_1}$ implies that $R$ contains a tuple 
$t_2$ satisfying $(t_2[3],t_2[4]) \in O_1$. Since $R$ entails $(O_1(x_1, x_2) \implies \uuOtwo(x_3, x_4))$
we have that $(t_2[1], t_2[2]) \notin O_1$. It follows that $R$ is a $[(O_1(x_1, x_2) \implies \uuOtwo(x_3, x_4))]$-relation and contradicts the assumption that $\mathbb{A}$ does not pp-define 
such relations. Thus, every $\psi^{O_1}$ is satisfiable.

It follows that for every atomic formula $\psi^{O_1}$ in $\varphi_I^{O_1}$ the projection of $\psi^{O_1}$ to any $\{v_i,v_j\}$ with $\instance_{i,j} \ni O_1$ equals $O_1$.
Observe that the following is also true: the projection of every $\psi^{O_1}$ to $\{ v_i, v_j \}$ with $O_2 \subseteq 
\instance_{i,j} \subseteq \uuOtwo$ contains $O_2$. Indeed, otherwise there exists a minimal 
$\mathbb{I}^{O_1}$
of $\{ (i,j) \mid \instance_{i,j} \ni O_1 \}$ and $(i_3, i_4)$ with $\instance_{i_3, i_4}$ satisfying 
$O_2 \subseteq P_{i_3, i_4} \subseteq \uuOtwo$ such that 
$(\psi \wedge \bigwedge_{(i,j) \in \mathbb{I}^{O_1}} O_1(v_i,v_j))$ entails $(v_{i_3} = v_{i_4})$.
Let $(i_1, i_2) \in \mathbb{I}^{O_1}$ and $\{ y_1, \ldots, y_l \}$ be all variables occurring in $\psi$.
In the same way as above,  we argue that 
$R(x_1, x_2, x_3, x_4)$ defined in the same way as the relation in~(\ref{eq:quaternary01impO2})
contains a tuple $t_1$ such that 
$(t_1[1], t_1[2]) \in O_1$ and $(t_1[3] = t_1[4])$ as well as a tuple $t_2$ satisfying
$(t_2[1], t_2[2]) \in \uuOtwo$ and $(t_1[3],  t_1[4]) \in O_2$.
It implies that $R$
efficiently entails $(O_1(x_1, x_2) \rightarrow x_3 = x_4)$, which contradicts the assumption that $R$
does not pp-define $[(O_1(x_1, x_2) \rightarrow x_3 = x_4)]$-relations.

In the second step of the transformation of the instance, we obtain $\varphi_I^{O_2}$ from $\varphi_I^{O_1}$ by replacing every atomic formula
$\psi^{O_1}$ with $\psi^{O_2}$ equivalent to $(\psi^{O_1} \wedge \eta_S^{O_2})$
where $\eta_S^{O_2}$ is $\varphi_S^{O_2}$ restricted to variables in $\psi$. 
We claim that a relation defined by any $(\psi^{O_1} \wedge \eta_S^{O_2})$  is non-empty. 
Assume on the contrary that it is not the case for some $(\psi^{O_1} \wedge \eta_S^{O_2})$. Define $\mathbb{I}^{O_2}$ to be the set of all $( i,j ) \in [n]^2$ where $\{ v_i , v_j \}$ are in the scope of $\psi^{O_2}$ and the projection of $\psi^{O_1}$
to $\{ v_i, v_j \}$ contains $\uuOtwo$. Since $\psi^{O_1}$  is satisfiable and $\psi^{O_2}$ is not and since the projection of $\psi^{O_1}$ to $\{ v_i, v_j \}$ with $\instance_{i,j}$ satisfying 
$O_2 \subseteq \instance_{i,j} \subseteq \uuOtwo$ contains $O_2$, there exists a minimal non-empty $\mathbb{J} \subseteq \mathbb{I}^{O_2}$ and  $(i_3, i_4) \in \mathbb{I}^{O_2} \setminus \mathbb{J}$
such that $(\psi^{O_1} \wedge \bigwedge_{(i,j) \in \mathbb{J}} O_2(v_i, v_j))$
entails $(v_{i_3} = v_{i_4})$. This time, having in mind that $\psi^{O_1}$ is equivalent to 
$\psi \wedge \eta_S^{O_1}$ 
we can pp-define a $[(O_2(x_1, x_2) \implies x_3 = x_4), 
(\uuOtwo(x_1, x_2) \wedge \uuOtwo(x_3, x_4))]$-relation, which again contradicts the assumption
that such relations are not pp-definable in $\mathbb{A}$.

Observe that in the obtained $\varphi_I^{O_2}$ 
a projection of any constraint $\psi^{O_2}$ to any pair of variables $\{ v_i, v_j \}$ in the scope of $\psi$
equals $E, N$ or $=$ and the projections agree in different atomic formlas Thus, in fact $\mathcal{I}'$ obtained from
$\varphi_I^{O_2}$ by replacing atomic formulas of the form $R(v_{i_1}, \ldots, v_{i_r})$ with constraints $((v_{i_1}, \ldots, v_{i_r}), R)$
is a $(2, \maxbound)$-minimal instance equivalent to an instance of $\csp((D; E, N, =))$
where $(D; E)$ is $\homograph$. By Observation~\ref{obs:ENEqualityWidth}, it follows that $\mathcal{I}'$, and hence, by Proposition~\ref{prop:minimality}, $\mathcal{I}$ has a solution.
$\square$.

\section{Proof of Lemma~\ref{lem:hnear23min}}
Let $\mathcal{I}$ be a $(2,3)$-minimal non-trivial instance of the CSP
equivalent to an instance of $\csp(\mathbb{A})$ and let $\varphi_I$ be a conjunction of atomic formulas obtained from $\mathcal{I}$ by replacing every
constraint of the form $((v_{i_1}, \ldots, v_{i_r}), R)$ with an atomic formula
$R(v_{i_1}, \ldots, v_{i_r})$.
We need to prove that $\varphi_I$
is satisfiable. 
For an atomic formula $\psi$ in $\varphi_I$ we define  $\psi^N$ to be  equivalent to 
$\psi \wedge \eta^N$ where  $\eta^N$ is the conjunction of $\bigwedge_{\instance_{i,j} \supseteq \{ N \}} N(v_i, v_j)$ restricted to variables occurring in $\psi$. 
We first show that 
$\varphi_I^N$ obtained from $\varphi_I$ by replacing every  $\psi$ with 
$\psi^N$ is satisfiable.

Assume first on the contrary that there exists  $\psi^N$ in  $\varphi_I^N$ defining an empty relation. Then there exists a non-empty and minimal set $\mathbb{I}^N$ of pairs
of indices  $(i,j)$ satisfying $\instance_{i,j} \supseteq \{ N \}$ and $(i_3, i_4) \notin \mathbb{I}^N$
but with  $P_{i_3,i_4} \supseteq \{ N \}$ such that $(\psi \wedge \bigwedge_{(i,j) \in \mathbb{I}^N} N(v_i, v_j))$
entails $\uuE(v_{i_3}, v_{i_4})$. Let $\{ y_1, \ldots, y_l\}$ be all variables occuring in $\psi$ and $(i_1, i_2) \in \mathbb{I}^N$. Then:
\begin{align}
\label{eq:quaternaryNC2omega}
R(x_1, x_2, x_3, x_4) &\equiv (\exists y_1 \cdots \exists y_l~(\psi \wedge \bigwedge_{(i,j) \in \mathbb{I}^N \setminus \{ (i_1, i_2) \}}
 N(v_i, v_j) \wedge \bigwedge_{j \in [4]} x_j = v_{i_j})) 
\end{align} 
entails $(N(x_1, x_2) \implies \uuE(x_3, x_4))$. Since $\mathbb{I}$ is minimal and non-empty, the relation $R$
contains a tuple $t_1$ such that $(t_1[1], t_1[2]) \in N$ and $(t_1[3], t_1[4]) \in \uuE$ 
as well as a tuple $t_2$ satisfying $(t_2[1], t_2[2]) \in \uuE$ and $(t_2[3], t_2[4]) \in N$. It follows that $\mathbb{A}$ pp-defines a $[(N(x_1, x_2) \implies \uuE(x_3, x_4))]$-relation, which
contradicts the assumption of the lemma. 

We will now show that $\mathcal{I}^N$ obtained from $\varphi_I^N$ 
by replacing every atomic formula of the form $R(v_{i_1}, \ldots, v_{i_r})$ by
$((v_{i_1}, \ldots, v_{i_r}), R)$ 
is $(2,3)$-minimal. Since $\mathcal{I}$ is $(2,3)$-minimal, clearly every subset of at most $3$ variables is in the scope of some constraint. 
Further, if the projection of any constraint in $\mathcal{I}^N$  to $\{ v_i, v_j \}$ with $i,j \in [n]$ contains $N$,
it equals $N$, and hence the projections of any two constraint in $\mathcal{I}^N$
to $\{ v_i, v_j \}$ with $\instance_{i,j} \supseteq \{ N \}$ and $i,j \in [n]$ agree.
   In order to complete the  proof of the fact that $\mathcal{I}^N$ satisfy Property~\ref{M2} in Definition~\ref{def:minimality},  we concentrate
   on $\instance_{i,j}$ that do not contain $N$. For all such $i,j \in [n]$  
and all $O \in \instance_{i,j}$  we will show that $(\psi^N \wedge O(v_{i},v_{j}))$ is satisfiable. 
If $\left| \instance_{i,j} \right| = 1$, we are done by the fact that every $\psi^N$ is satisfiable.
Thus, we turn to the case where $\instance_{i,j} = \{ E, = \}$.
Assume on the contrary that there are $\psi^N$,  $\{ v_{i_3}, v_{i_4} \}$ in the scope of $\psi^N$ and $O \in \{ E, = \}$
such that $(\psi^{N} \wedge O(v_{i_3},v_{i_4}))$ is not satisfiable. It follows that 
there is a non-empty and minimal $\mathbb{I}^N \subseteq \{ (i,j) \mid \instance_{i,j} \supseteq \{ N \} \}$ such that 
$(\psi \wedge \bigwedge_{(i,j) \in \mathbb{I}^N} N(v_i, v_j))$ entails $O_1(v_{i_3}, v_{i_4})$ where $O_1 \in \{ E, = \} \setminus \{ O \}$. Hence
$R(x_1, x_2, x_3, x_4)$ defined by the formula~(\ref{eq:quaternaryNC2omega})
entails $(N(x_1, x_2) \implies O_1(x_3, x_4))$. Again, using the minimality of $\mathbb{I}^N$
we can show that $R$ contains a tuple $t_1$ such that $(t_1[1], t_1[2]) \in N$ and $(t_1[3], t_1[4]) \in O_1$
as well as a tuple $t_2$ satisfying $(t_2[1], t_2[2]) \in \uuE$ and $(t_2[3], t_2[4]) \in O$.
It follows that $R$ efficiently entails $(N(x_1, x_2) \implies O_1(x_3, x_4))$ and that it entails $(\uuE(x_3, x_4))$.
Since $\mathbb{A}$ contains $N$, the relation $R$ is a $[(N(x_1, x_2) \implies O_1(x_3, x_4)),(\uuE(x_3, x_4))]$-relation pp-definable in $\mathbb{A}$.
Now, regardless of whether $O_1$ is $E$ or $=$, we have a contradiction with the assumptions of the lemma.

From now on we assume that $\mathcal{I}^N$ is $(2,3)$-minimal and we set $P^N_{i,j}$ for $i,j \in [n]$ to be the set of orbitals $O \in \{ E, N, = \}$ contained in 
the projection of any constraint in $\mathcal{I}^N$ to $\{ v_i, v_j \}$. 
Actually $P^N_{i,j} = \instance_{i,j}$ if $N \notin \instance_{i,j}$
and $P^N_{i,j} = \{ N \}$ otherwise. 
 If $\left| P^N_{i,j} \right| = 1$ for all $i,j \in [n]$, then 
$\mathcal{I}^N$ is a non-trivial $(2,3)$-minimal instance of the CSP equivalent to an instance of $\csp((D;E, N, =))$ where $(D;E)$ is $C^2_{\omega}$.
Since $\mathbb{L}_{C^2_{\omega}}$ is $3$,
it follows by Observation~\ref{obs:ENEqualityWidth} that $\mathcal{I}^N$ and, by Proposition~\ref{prop:minimality}, that $\mathcal{I}$ has a solution. It completes the proof of the lemma 
in the case where $\left| P^N_{i,j} \right| = 1$ for all $i,j \in [n]$. 

From now on we assume that there exists $i,j \in [n]$ such that $\left| P^N_{i,j} \right| = 2$. In particular, we have that $\uuE$ is definable in $\mathbb{A}$.
Since $\mathcal{I}^N$ is $(2,3)$-minimal, it is easy to see that $\instance_{i,j}^N$ with $i,j \in [n]$
respects the transitivity of $\uuE$ in the sense that for all  $i,j,k \in [n]$: if $P^N_{i,j} \cap \{ E, = \} \neq \emptyset$ and 
$P^N_{j,k} \cap \{ E, = \} \neq \emptyset$, then $P^N_{i,k} \cap \{ E, = \} \neq \emptyset$.  Set now $J_1, \ldots, J_m$
to be the partition of $\V$ such that for all $l \in [m]$ and $i,j \in J_l$ we have that $P^N_{i,j} \cap \{ E, = \} \neq \emptyset$. 
Then construct $\varphi^i$  
for all $i \in [m]$ to be a formula that for every atomic formula $\psi^N$ in $\varphi_I^N$ contains 
 $\psi^i$ equivalent to $(\exists v_{j_1} \cdots \exists v_{j_l}~\psi^N)$ where 
 $\{ j_1, \ldots, j_l \} = [n] \setminus J_i$.
Observe that every $\varphi^i$ is equivalent to an instance $\mathcal{I}^i$  of  $\csp(\mathbb{A}_i)$ such that $\mathbb{A}_i$ has a pp-definition in $\mathbb{A}$
and hence is also preserved by an oligopotent qnu-operation $f$. Furthermore, the core $\mathbb{A}^c_i$ of $\mathbb{A}_i$ for every $i \in [m]$
is finite. Assume without loss of generality that $\mathbb{A}^c_i$ admits a fo-definition over a single edge $\{ a_{2i-1}, a_{2i} \}$ in $A$ and that
$f$ is idempotent on $\{ a_1, \ldots, a_{2m}\}$.
Let $h_i$ with $i \in [m]$ 
be a homomorphism from $\mathbb{A}_i$
to $\mathbb{A}^c_i$ that is, without loss of generality, idempotent on $\{ a_{2i-1}, a_{2i} \}$.
Then, since $f$ preserves $\uuE$, we have that $h_i(f(x_1, \ldots, x_k))$ restricted to $\{ a_{2i-1}, a_{2i} \}$ is
a near-unanimity operation preserving $\mathbb{A}^c_i$. Since an instance of the CSP equivalent to $\varphi^i$ is $(2,3)$-minimal also when evaluated in $\mathbb{A}^c_i$,
it follows by Theorem~10 that $\varphi^i$ for all $i \in [l]$ has a solution $s_i$ in 
$\{ a_{2i-1}, a_{2i} \}$.

For all $l \in [m]$ and $i,j \in J_l$ we set $O_{i,j}$ to $E$ if 
$s(v_i) \neq s(v_j)$ and to $=$ otherwise. In order to complete the proof of the lemma we will show that 
$(\varphi^N_I \wedge \bigwedge_{l \in [m]} \bigwedge_{i,j \in J_l} O_{i,j}(v_i, v_j))$ is satisfiable.
Assume the contrary. Then there exists a 
$\psi^N$ in $\varphi_I^N$ and a
minimal non-empty subset $\mathbb{J}$ of pairs of indices $(i,j) \in \bigcup_{l \in [m]} J_l \times J_l$ of variables occurring in $\psi^N$ such that $(\psi^N \wedge \bigwedge_{(i,j) \in \mathbb{J}} O_{i,j}(v_i, v_j))
$ is not satisfiable.
Since $s_i$ is a solution to $\varphi^i$ for all $i \in [m]$ we have that $\mathbb{J}$ contains at least two pairs: $(i_1, i_2) \in J_a$ and $(i_3, i_4) \in J_b$
with $a \neq b$. Set $y_1, \ldots, y_l$ to be all variables occurring in  $\psi$ and consider now the relation
\begin{align}
R(x_1, x_2, x_3, x_4) &\equiv (\exists y_1 \cdots y_m~(\psi^N \wedge \bigwedge_{(i,j) \in \mathbb{J} \setminus \{ (i_1, i_2), (i_3, i_4) \}} 
O_{i,j}(v_i, v_j) \wedge \bigwedge_{k \in [4]} x_k = v_{i_k})). \nonumber
\end{align}
Since $a \neq b$, the relation $R$ entails 
$N(x_2, x_3)$ and $(O_{i_1, i_2}(x_1, x_2) \implies O'_{i_3, i_4}(x_3, x_4))$ where 
$O'_{i_3, i_4} \in \{ E, = \} \setminus \{ O_{i_3, i_4} \}$. By the minimality of $\mathbb{J}$ 
we have that $R$ contains a tuple $t_1$ 
such that $(t_1[1], t_1[2]) \in O_{i_1, i_2}$ and $(t_1[3], t_1[4]) \in O'_{i_3, i_4}$ as well as a tuple
$t_2$ satisfying $(t_2[1], t_2[2]) \in O'_{i_1, i_2}$ where $O'_{i_1, i_2} \in \{ E, = \} \setminus \{ O_{i_1, i_2} \}$ and $(t_2[3], t_2[4]) \in O_{i_3, i_4}$.
Since $\mathbb{A}$ contains $E,N$,  we have that it  pp-defines a $[(O_{i_1, i_2}(x_1, x_2) \implies O'_{i_3, i_4}(x_1, x_2), (\uuE(x_1, x_2) \wedge N(x_2, x_3) \wedge 
\uuE(x_3, x_4))]$-relation where $O_{i_1, i_2}, O'_{i_3, i_4} \in \{ E, = \}$. It contradicts the assumptions of the lemma and proves that $s$ satisfies  $\varphi_I$, and hence it is a solution to $\mathcal{I}$. It completes the proof of the lemma.
$\square$

\section{Proof of Lemma~\ref{lem:Comega23excludes}}

Let $\mathcal{I}$ be a non-trivial $(2,3)$-minimal instance of the CSP equivalent to 
an instance of
$\csp(\mathbb{A})$ and let 
$0,1 \in A$ be two elements in two
different infinite cliques in $C^{\omega}_2$ and $g: A \rightarrow \{ 0,  1\}$  a function that sends all elements in  the infinite clique containing  $0$ to $0$ and all elements in the infinite clique containing $1$ to $1$.  Define the structure $g(\mathbb{A})$ over the domain $\{ 0,1 \}$ to be a structure in which every $R$ in $\mathbb{A}$ is replaced by $g(R)$ 
in which in turn every tuple $t \in R$ is replaced by $g(t)$. Since $f$ preserves $E$ and hence 
$\uuE(x_1, x_2) \equiv (\exists y~E(x_1, y) \wedge E(y, x_2))$, we may assume without loss of generality that $f$ sends 
$(a_1, \ldots, a_k)$ for $a_1, \ldots, a_k$ in the same infinite clique to an element in the same infinite clique.
It follows that 
$f'(x_1, \ldots, x_k) := g(f(x_1, \ldots, x_k))$ restricted to $\{ 0, 1 \}$ is a near-unanimity operation preserving $g(\mathbb{A})$. Consider then $g(\mathcal{I})$ which is an instance of $\csp(g(\mathbb{A}))$
obtained from $\mathcal{I}$ by replacing every constraint $((v_{i_1}, \ldots, v_{i_r}),R)$  by $((v_{i_1}, \ldots, v_{i_r}),g(R))$.
Notice that $g(\mathcal{I})$ is $(2,3)$-consistent. Clearly for every set of three variables, there is a constraint containing these variables in its scope. Further, for any two variables $v_i, v_j$ with $i,j \in [n]$, the projection of any constraint $$((v_{k_1}, \ldots, v_{k_r}),g(R))$$ to $\{v_i,  v_j \}$ contains $\{ (0,0), (1, 1) \}$ if the projection of $((v_{k_1}, \ldots, v_{k_r}),R)$ to $\{ v_i, v_j \}$ contains $=$ or $E$ and $\{ (0,1), (1, 0) \}$ if the projection of $((v_{k_1}, \ldots, v_{k_r}),R)$ to $\{ v_i, v_j \}$ contains $N$. Hence $g(\mathcal{I})$ is $(2,3)$-minimal. By Theorem~10 we have that
$g(\mathcal{I})$ has a solution $s_g$. 
For every constraint $C$ we let
\begin{itemize}
\item $\mathbb{I}_C^d$ to be the set of pairs of indices $(i,j) \in [n]^2$ of variables occurring in $C$ such that
 $s_g(v_i) \neq s_g(v_j)$
and 
\item $\mathbb{I}_C^s$ to be the set of pairs of indices $(i,j) \in [n]^2$ of variables occurring in $C$
such that
 $s_g(v_i) = s_g(v_j)$ and $\instance_{i,j} \supseteq \{ E \}$.
\end{itemize}

Consider now $\mathcal{I}^N$ obtained from $\mathcal{I}$ by replacing every constraint $C := ((v_{i_1}, \ldots, v_{i_r}), R)$ in $\mathcal{I}$
by $C^N := ((v_{i_1}, \ldots, v_{i_r}), R')$ 
where $R'(v_{i_1}, \ldots, v_{i_r}) \equiv (R(v_{i_1}, \ldots, v_{i_r}) \wedge \bigwedge_{(i,j) \in \mathbb{I}_\psi^d} N(v_i, v_j))$. Since $s_g$ is the solution to $g(\mathcal{I})$, we have that every $R'$ is non-empty. Furthermore, we claim that the projection of every $C^N$ to $\{ v_i, v_j \}$ with $(i,j) \in \mathbb{I}_C^s$ contains $E$. Assume on the contrary that there is a constraint $C^N$ and a pair 
$(i_3, i_4) \in \mathbb{I}_C^s$ such that the projection of $C^N$ to $\{ v_i, v_j \}$ does not contain $E$.
It follows that there exists a minimal non-empty subset $\mathbb{J}^d_{C} \subseteq \mathbb{I}^d_{C}$
 such that $(\eta_{\psi} := \bigwedge_{(i,j) \in \mathbb{J}^C_{\psi}} N(v_i, v_j) \wedge R')$ entails $\uuN(v_{i_3}, v_{i_4})$ and does not entail $N(v_{i_3}, v_{i_4})$. 
Since $\mathbb{A}$ does not pp-define $\uuN$, we have that $\eta_{\psi}$ entails $(v_{i_3} = v_{i_4})$.
Let $(i_1, i_2) \in \mathbb{J}^C_{\psi}$ and $\{ y_1, \ldots, y_l \}$ be the scope of $C$. Then 
\begin{align}
R(x_1, x_2, x_3, x_4) &:= (\exists y_1 \cdots \exists y_l~(\bigwedge_{(i,j) \in \mathbb{J}^d_{C} \setminus \{ (i_1, i_2) \}} N(v_i, v_j) \wedge \psi \wedge \bigwedge_{j \in [4]} v_{i_j} = x_j)) \nonumber
\end{align}
entails $(N(x_1, x_2) \implies (x_3 = x_4))$. Since $s_g$ is a solution to $g(\mathcal{I})$, the relation $R$
clearly contains a tuple $t_1$ with $(t_1[1], t_1[2]) \in N$ and $(t_1[3] = t_1[4])$. On the other hand, the set 
$ \mathbb{J}^d_{C}$ is minimal, and hence $R$ contains a tuple $t_2$ satisfying $(t_2[3] \neq t_2[4])$ and 
$(t_2[1], t_2[2]) \notin N$.  We obtain that $R$ is a $[(N(x_1, x_2) \rightarrow (x_3 = x_4))]$-relation.
Since $\mathbb{A}$ contains $N$, it pp-defines $R$.
It contradicts the assumption of the lemma. We have that the projection of every $C^N$ to any $(i,j ) \in \mathbb{I}_C^s$ contains $E$. 

In the next step we construct $\mathcal{I}^E$ which is obtained from $\mathcal{I}^N$ by replacing every constraint $C^N := ((v_{k_1}, \ldots, v_{k_r}), R)$ by $C^E := ((v_{k_1}, \ldots, v_{k_r}), R')$
where $R'(v_{k_1}, \ldots, v_{k_r}) \equiv (R(v_{k_1}, \ldots, v_{k_r}) \wedge \bigwedge_{(i,j) \in \mathbb{I}_C^s} E(v_i, v_j))$.
We claim that every such $R'$ is non-empty. Assume the contrary. Then there exists a minimal 
and non-empty $\mathbb{J}_C^s \subseteq \mathbb{I}_C^s$ and $(i_3, i_4) \in 
\mathbb{I}_C^s \setminus \mathbb{J}_C^s$
 such that  
$(\eta_{C}^E := R'(v_{k_1}, \ldots, v_{k_r}) \wedge \bigwedge_{(i,j) \in \mathbb{J}_C^s} E(v_i, v_j))$ entails $\uuN(v_{i_3}, v_{i_4})$ and does not entail $N(v_{i_3}, v_{i_4})$.
Again, since $\mathbb{A}$ does not pp-define $\uuN$, the formula $\eta_{\psi}^E$ entails $(v_{i_3} = v_{i_4})$.
Let $(i_1, i_2) \in \mathbb{J}_{C}^s$.
 Then 
\begin{align}
R(x_1, x_2, x_3, x_4) &:= (\exists v_{k_1} \cdots \exists v_{k_r}~(\psi^N \wedge \bigwedge_{(i,j) \in \mathbb{J}_C^s \setminus \{ (i_1, i_2) \}} E(v_i, v_j) \wedge \bigwedge_{j \in [4]} v_{i_j} = x_j)) \nonumber
\end{align} 
entails $(E(v_{i_1}, v_{i_2}) \implies (v_{i_3} = v_{i_4}))$. Since the projection of every $C^N$ to any
$(i,j ) \in \mathbb{I}_C^s$ contains $E$, the relation $R$ contains a tuple $t_1$ with $(t_1[1], t_1[2]) \in E$
and $(t_1[3] = t_1[4])$.
On the other hand, since $\mathbb{J}_C^s$ is minimal there is a tuple $t_2 \in R$ with $(t_2[3] \neq t_2[4])$
and $(t_2[1], t_2[2]) \notin E$. Thus, $R$ is a $[(E(x_1, x_2) \rightarrow x_3 = x_4)]$-relation. Since $\mathbb{A}$ contains $E$ and $N$, we have that $R$
is pp-definable in $\mathbb{A}$. It contradicts the assumption of the lemma and completes the proof of the fact that every $R'$ is non-empty.

Since $\mathbb{A}$ does not pp-define $\uuN$, for every $i,j \in [n]$, the projection of any $C^E$ 
to $\{ v_i , v_j \}$ is $N$ if $(i,j) \in \mathbb{I}_C^d$, it is $E$ if $(i,j) \in \mathbb{I}_C^s$ and 
$=$ otherwise that is when $\instance_{i,j} = \{ = \}$. Thus, $\mathcal{I}^E$ is a $(2,3)$-minimal instance equivalent 
to an instance of $\csp((D; E, N, =))$, where $(D; E)$ is $C^{\omega}_2$. Since $\mathbb{L}_{C^{\omega}_2}$ is $3,$ it follows by Observation~\ref{obs:ENEqualityWidth} that $\mathcal{I}^E$ and in consequence, by Proposition~\ref{prop:minimality}, $\mathcal{I}$ has a solution.
It completes the proof of the lemma.
$\square$

\section{Proof Of Lemma~\ref{lem:graphreductbinary}}

We will show that in any of the considered cases $\mathbb{A}$ pp-defines none of the relations mentioned in the
formulation of Lemma~\ref{lem:injO1dominating} where $O_1 = E, O_2 = N$ 
in Cases~\ref{graphreductbinary:max} and~\ref{graphreductbinary:Econstant} and $O_1 = N, O_2 = E$
in Cases~\ref{graphreductbinary:min} and~\ref{graphreductbinary:Nconstant}.

Assume first on the contrary that $\mathbb{A}$ pp-defines a $[(O_1(x_1, x_2) \implies \uuOtwo(x_3, x_4))]$-relation.
The relation 
$R$ contains:
\begin{itemize}
\item a tuple $t_1$ such that $(t_1[1],t_1[2]) \in O_1$ and $(t_1[3], t_1[4]) \in \uuOtwo$, and
\item a tuple $t_2$ such that $(t_2[1],t_2[2]) \in \uuOtwo$ and $(t_1[3], t_1[4]) \in O_1$.
\end{itemize} 
Observe that $t:= f(t_1, t_2)$, where $f$ is any operation from the formulation of the lemma, satisfies $(t[1],t[2]) \in O_1$ and $(t[3], t[4]) \in O_1$ in any case. 
But containment of $t \in R$ contradicts the fact that $R$ 
efficiently entails $(O_1(x_1, x_2) \implies \uuOtwo(x_3,x_4))$. 

Then we turn to showing that $[(O_1(x_1, x_2) \implies x_3 = x_4)]$-relations
$R$ cannot be pp-definable in $\mathbb{A}$. Assume on the contrary that the relation $R$
contains 
\begin{itemize}
\item a tuple $t_1$ such that $(t_1[1],t_1[2]) \in O_1$ and $t_1[3] = t_1[4]$, and
\item a tuple $t_2$ such that $(t_2[1],t_2[2]) \in \uuOtwo$ and $(t_1[3] \neq t_1[4])$.
\end{itemize} 
This time, regardless of the case,  $t:= f(t_1, t_2)$ satisfies $(t[1],t[2]) \in O_1$ and $(t[3] \neq t[4])$.
It contradicts the assumption that $R$ 
efficiently entails $(O_1(x_1, x_2) \implies x_3  = x_4)$.

Finally we turn to $[(O_2(x_1,x_2) \implies x_3 = x_4), (\uuOtwo(x_1, x_2) \wedge \uuOtwo(x_3, x_4))]$-relations.
Assume on the contrary that $R$ is such a relation and is preserved by one of operations $f$
mentioned in the formulation of the lemma. 
The relation $R$  contains:
\begin{itemize}
\item a tuple $t_1$ such that $(t_1[1], t_1[2]) \in O_2$ and $t_1[3] = t_1[4]$, and
\item a tuple $t_2$ such that $(t_2[1] = t_2[2])$ and $(t_2[3], t_2[4]) \in O_2$.
\end{itemize}
If $f$ is of behaviour min or max and $O_1$-dominated or  $O_1$-constant, 
then clearly $\uuOtwo$ and in consequence $R$ are not preserved by $f$.
If $f$ is of behaviour min or max and balanced, then $t:= f(t_1, t_2)$ satisfies $(t[1],t[2]) \in O_1$ and $(t[3] \neq t[4])$.
This contradicts the fact that $R$ efficiently entails 
$(O_2(x_1,x_2) \implies x_3 = x_4)$ and completes the proof of the lemma. 
$\square$

\section{Proof of Lemma~\ref{lem:graphmajorityequalities}}

A $[(O(x_1, x_2) \implies x_3 = x_4)]$-relation $R$
  contains:
\begin{itemize}
    \item a tuple $t_1$ such that $(t_1[1], t_1[2]) \in O$ and
    $(t_1[3] = t_1[4])$, and
    \item a tuple $t_2$ such that $(t_2[1], t_2[2]) \notin O$ and
    $(t_2[3] \neq t_2[4])$, 
\end{itemize}
Let $f$ be a ternary injections of behaviour majority specified in  the formulation of the lemma.
To reach a contradiction, it is enough to show that $f$ applied, 
perhaps multiple times, to $t_1, t_2$ produces a tuple $t$ such that $(t[1],t[2]) \in O$ and $(t[3] \neq t[4])$.
Indeed, it implies that $R$ does not entail 
$(O(x_1, x_2) \implies x_3 = x_4)$ and contradicts the assumption.

Let $t := f(t_1, t_1, t_2)$. 
Since $f$ is an injection and $t_2[3] \neq t_2[4]$, it always holds that
$t[3] \neq t[4]$. If $t_2[1] \neq t_2[2]$, then since $f$ is of behaviour majority,
we have that $(t[1], t[2]) \in O$. If $t_2[1] = t_2[2]$, then  since    $(t_1[1], t_1[2]) \in O$, it is straightforward to check that 
$(t[1], t[2]) \in O$ for $f$ that is 
\begin{itemize}
    \item hyperplanely balanced and of behaviour projection,
    \item hyperplanely of behaviour max and $E$-dominated,   
    \item hyperplanely of behaviour min and $N$-dominated,
    \item hyperplanely $E$-constant if $O = E$,
    \item hyperplanely $N$-constant if $O = N$.
\end{itemize}

For the remaining two cases: the case where 
 $f$ is hyperplanely $E$-constant and $O$ is $N$ 
 and the case where $f$ is hyperplanely $N$-constant and $O$
 is $E$,
observe that  $t' = f(t_1, t_1, t)$ satisifes $(t'[1],t'[2]) \in O, t'[3] \neq t'[4]$.
It completes the proof of the lemma.
$\square$

\section{Proof of Lemma~\ref{lem:graphminorityequalities}}

Assume on the contrary that $\mathbb{A}$ pp-defines a
$[(O(x_1, x_2) \implies x_3 = x_4)]$-relation $R$.
The relation $R$ contains
\begin{itemize}
    \item a tuple $t_1$ such that $(t_1[1], t_1[2]) \in O$ and $t_1[3] = t_1[4]$, and
    \item a tuple $t_2$ such that $(t_2[1], t_2[2]) \notin O$ and $t_1[3] \neq t_1[4]$.
\end{itemize}

The proof goes along the lines of the proof of Lemma~\ref{lem:graphmajorityequalities}.
Again, to reach a contradiction, we will produce
 a tuple $t$ satisfying $(t[1], t[2]) \in O$ and $(t[3] \neq t[4])$.
It contradicts that $R$ is a $[(O(x_1, x_2) \implies x_3 = x_4)]$-relation and will complete the proof of the lemma. 

Consider $t = f(t_1, t_2, t_2)$ where $f$ is any operation from the formulation of the lemma. 
Since $t_2[3] \neq t_2[4]$ and $f$ is an injection, we have that 
$t[3] \neq t[4]$. Furthermore, since $f$ is of behaviour minority
we have $(t[1], t[2]) \in O$ always when 
$(t_2[1] \neq t_2[2])$.
Thus, we turn to the case where $t_2[1] = t_2[2]$.
If $f$ is hyperplanely balanced and of behaviour projection, then
$(t[1], t[2]) \in  O$ and we are done.
Further, if $f$ is hyperplanely of behaviour projection and E-dominated, or
hyperplanely balanced of behaviour xnor we obtain 
$(t[1], t[2]) \in  E$ and hence we are done in the case where $O$ is $E$.
Similarly, if $f$ is hyperplanely of behaviour projection and N-dominated, or
hyperplanely balanced of behaviour xor we obtain 
$(t[1], t[2]) \in  N$ and in this case we are done where $O$ is $N$.
Thus, either we are done or we have $t$ such that $(t[1], t[2]) \in O_1$ with $\{ O_1 \} = \{ E, N \} \setminus \{ O \}$ and 
$(t[3] \neq t[4])$. Now $t' := f(t_1, t, t)$ is the desired tuple satisfying $(t'[1],  t'[2]) \in O$ and $(t'[3] \neq t'[4])$.
It completes the proof of the lemma.
$\square$

\section{Proof of Lemma~\ref{lem:implexcludbyh}}

Assume first on the contrary that a $[(N(x_1,x_2) \implies \uuE(x_3, x_4))]$-relation $R$ is pp-definable in $\mathbb{A}$
and let $t_1$ be a tuple in $R$ with $(t_1[1], t_1[2]) \in N$ and $(t_1[3], t_1[4]) \in \uuE$ and let
$t_2 \in R$ satisfy $(t_2[1], t_2[2]) \in \uuE$ and $(t_2[3], t_2[4]) \in N$. Then for $t = h(t_1, t_1, t_2)$
we have that $(t[1], t[2]), (t[3], t[4]) \in N$.
It contradicts the fact that $R$ efficiently entails $(N(x_1,x_2) \implies \uuE(x_3, x_4))$ and completes the proof for the first case.

In the second case we again assume the contary and consider a $[(N(x_1,x_2) \implies x_3 = x_4)]$-relation $R$.
This time we have $t_1 \in R$ with $(t_1[1],t_1[2]) \in N$ and $(t_1[3] = t_1[4])$  and $t_2 \in R$
with $(t_2[1], t_2[2]) \in \uuE$ and $(t_2[3] \neq t_2[4])$. Now, 
 regardless of whether $(t_2[3], t_2[4]) \in N$ or $(t_2[3], t_2[4]) \in E$
the tuple $t = h(t_1, t_1, t_2)$ satisfies $(t[1], t[2]) \in N$ and $(t[3] \neq t_4)$.
It contradicts the fact that $R$
is a $[(N(x_1,x_2) \implies x_3 = x_4))]$-relation. 

The proof of the third case goes in the same way. We assume on the contrary that $R$  pp-definable in $\mathbb{A}$
is a  $[(N(x_1,x_2) \implies E(x_3, x_4)), (\uuE(x_3, x_4))]$-relation. This time we have 
$t_1$ with $(t_1[1], t_1[2]) \in N$ and $(t_1[3] t_1[4]) \in E$ and $(t_2[1], t_2[2]) \notin N$ and $(t_2[3] = t_2[4])$.
Observe that $t := f(t_1, t_1, t_2)$ satisfies $(t[1],  t[2]) \in N$ and $(t[3] = t[4])$, which contradicts that 
$R$ is a $[(N(x_1,x_2) \implies E(x_3, x_4)), (\uuE(x_3, x_4))]$-relation.
$\square$

\section{Proof of Lemma~\ref{lem:atmostoneENorNE}}

In what follows, a quatenary tuple  $t$ is called an $\OP$-tuple with $O,P \in \{ E, N, = \}$ if $(t[1], t[2]) \in O$ and $(t[3], t[4]) \in P$.
A tuple is called an constant tuple if all its entries have the same value, injective if all its entries are pairwise different.  

Assume now on the contrary that $\mathbb{A}$ preserved by a ternary injection of behaviour minority or majority and a $k$-ary oligopotent qnu-operation $f$
pp-defines both  a $[(E(x_1, x_2) \rightarrow \uuN(x_3, x_4))]$-relation $R_{\EN}$ and  a $[(N(x_1,x_2) \implies \uuE(x_3,x_4))]$-relation $R_{\NE}$.
Since by Lemmas~\ref{lem:graphmajorityequalities} and~\ref{lem:graphminorityequalities}, the structure $\mathbb{A}$ pp-defines no $[(O(x_1, x_2) \rightarrow x_3 = x_4)]$-relations with $O \in \{ E, N \}$,
the relation $R_{\EN}$ contains an $\EN$-tuple. On the other hand $R_{\EN}$ contains  a tuple 
$t$ such that $(t[1], t[2]) \in \uuN$ and $(t[3], t[4]) \in E$. This time by Lemmas~\ref{lem:graphmajorityequalities} and~\ref{lem:graphminorityequalities}, we have that
$R_{\EN}$ must have an $\NE$-tuple. In a similar way we show that $R_{\NE}$ contains both an $\EN$-tuple and a $\NE$-tuple.

In the remainder of the proof we first show that a relation $R_{\NE}$ (resp., $R_{\EN}$) pp-defines a 
$[(N(x_1, x_2) \rightarrow \uuE(x_3, x_4))]$-relation
$R'_{\NE}$ (resp. a $[(E(x_1, x_2) \rightarrow \uuN(x_3, x_4))]$-relation $R'_{\EN}$) that contains 
an   $\EN$-tuple and a  $\NE$-tuple of some desirable properties.
Finally, we will show that both $R'_{\EN}$ and $R'_{\NE}$ contain either an $\EE$-tuple or a $\NN$-tuple.
It follows that either $R'_{\EN}$ is not an $[(E(x_1, x_2) \rightarrow \uuN(x_3, x_4))]$-relation
or $R'_{\NE}$ is not a $[(N(x_1, x_2) \rightarrow \uuE(x_3, x_4))]$-relation. It will complete the proof of the lemma.
We continue with an easy observation on $\EN$-tuples and $\NE$-tuples.

\begin{observation}
\label{obs:oneequalityEN}
Let $t$ be either an $\EN$-tuple or a $\NE$-tuple. Then there are at most two different $i,j \in [4]$ such that $t[i] = t[j]$.
\end{observation}

\begin{proof}
The proof is performed by the inspection of cases. Since $t$ is an $\EN$-tuple or a $\NE$-tuple, we have
$(t[1], t[2]) \in O_1$, $(t[3], t[4]) \in O_2$ where $\{ O_1, O_2 \} = \{ E, N \}$. Hence, in particular $O_1, O_2$ are two different orbitals.
In the case where $t[2] = t[3]$, we have that 
$(t[1], t[3])\in O_1$ and
 $(t[2], t[4]) \in O_2$.  Since $O_1$ and $O_2$ are different, it follows that 
$t[1] \neq t[4]$.  
Furthermore, 
at most one may be true: either $(t[1] = t[3])$ or $(t[2] = t[4])$. Indeed, otherwise $(t[1], t[4])$ would be in both $O_1$ and $O_2$. If either $(t[1] = t[3])$ or $(t[2] = t[4])$, then since 
$(t[1] \neq t[2])$ and $(t[3] \neq t[4])$, we have that $t[1] \neq t[4]$.
The remaining case to be considered is when $t[1] = t[4]$. Then $(t[2], t[4]) \in O_1$ and $(t[1], t[3]) \in O_2$. It implies that $t[2] \neq t[3]$ and completes the proof of the lemma.  
\end{proof}

\noindent
We now prove that $R_{\EN}$ (resp. $R_{\NE}$) pp-defines $R'_{\EN}$ (resp. $R'_{\NE}$) of the desired properties.

\begin{observation}
\label{obs:ENtuples}
The relation $R_{\EN}$ (resp. $R_{\NE}$) pp-defines 
a $[(E(x_1, x_2) \implies \uuN(x_3, x_4))]$-relation $R'_{\EN}$ (resp. a $[(N(x_1, x_2) \implies \uuE(x_3, x_4) )]$ -relation $R'_{\NE}$) that contains 
\begin{enumerate}
\item \label{ENtuples:eqconnected} either both an  $\EN$-tuple $t_{\EN}$ and a $\NE$-tuple $t_{\NE}$ such that for all $t \in \{ t_{EN}, t_{NE} \}$
we have $(t[2] = t[3])$, or 
\item \label{ENtuples:injective} both an injective $\EN$-tuple $t_{\EN}$ and an injective $\NE$-tuple $t_{\NE}$.
\end{enumerate}
\end{observation}

\begin{proof}
Notice first that whenever $R \in \{ R_{\NE}, R_{\EN} \}$ contains  an $\OP$-tuple $t$ with $\{ O, P \} = \{ E, N \}$ which satisfies $t[i] \neq t[j]$ for some different $i,j \in [4]$, then
$R$ contains also a $\PO$-tuple $t'$ with $t'[i] \neq t'[j]$. Indeed, if $R$ is preserved by a ternary injection  $g$ of behaviour majority, then the desired $t'$ is the result of 
$g(t,t'',t'')$ where $t''$ is some $\PO$-tuple in $R$.  On the other hand, if $R$ is preserved by a ternary injection  $g$ of behaviour minority, the desired $t'$ is the result of 
$g(t,t,t'')$ where $t''$ is some $\PO$-tuple in $R$.

It follows, by Observation~\ref{obs:oneequalityEN} that either there is some exactly one $\{i,j\}$ such that all $\EN$-tuples and all $\NE$-tuples $t$ in $R$ satisfy $t[i] = t[j]$
or $R$ contains an injective $\EN$-tuple and an injective $\NE$-tuple. In the latter case we are done since for the desired $R'_{\NE}(R'_{\EN})$ we simply take $R_{\NE}(R_{\EN})$ while in the former case $\{ i,j \}$ is $\{ 2, 3 \}$, $\{ 1,3 \}$, $\{ 2,4 \}$, or $\{ 1, 4 \}$.
\begin{itemize}
\item If $\{ i,j \} = \{ 2,3 \}$, then we set $R'(x_1, x_2, x_3, x_4) \equiv R(x_1,x_2, x_3, x_4) $
\item If $\{ i,j \} = \{ 1,3 \}$, then we set  $R'(x_1,x_2, x_3, x_4) \equiv R(x_2, x_1, x_3, x_4)$.
\item If $\{ i,j \} = \{ 2,4 \}$, then we set  $R'(x_1,x_2, x_3, x_4) \equiv R(x_1, x_2, x_4, x_3)$.
\item If $\{ i,j \} = \{ 1,4 \}$, then we set  $R'(x_1,x_2, x_3, x_4) \equiv R(x_2, x_1, x_4, x_3)$.
\end{itemize}

It is straightforward to check that in each of the four cases $R'$ efficiently entails the formula $(E(x_1,x_2) \implies \uuN(x_3,x_4))$  (resp., $(N(x_1,x_2) \implies \uuE(x_3,x_4))$) if
$R$ efficiently entails $(E(x_1,x_2) \implies \uuN(x_3,x_4))$ (resp., $(N(x_1,x_2) \implies \uuE(x_3,x_4))$)
and contains both an $\EN$-tuple $t_{EN}$ and a $\NE$-tuple $t_{\EN}$ 
such that for all $t \in \{ t_{\EN}, t_{\NE} \}$ we have $(t[2] = t[3])$.
It completes the proof of the observation. 
\end{proof}

In Section~\ref{sect:canonicalop} we allowed to write 
$R_1 \cdots R_k(a_1, a_2)$  for $k$-tuples $a_1,a_2 \in A^k$ if and only if 
$R_i(a_1[i] , a_2[i])$ holds for all $i \in [k]$. 
With reference to this notation we say that $a_1, a_2$
are $\{ R_1, \ldots, R_k\}$-connected. Clearly, the relations $\{ R_1, \ldots, R_k\}$
do not have to be pairwise different. 
We write $n_{R_i(a_1,  a_2)}$ to denote the number of coordinates at which $R_j$ with $j \in [k]$ equals $R_i$. 
 For example, we write $N\!E\!\!=(a,b)$ for $a,b \in A^3$
if $(a[1],b[1]) \in N, (a[2],b[2]) \in N$, and $a[3] = b[3]$ and say that $a,b$ are $\{ E, N, = \}$-connected.  
In this case we have that $n_{E(a,b)} = n_{N(a,b)} = n_{=(a,b)} = 1$.

We say that a pair of tuples $u,s$ is constant-injective if $u$ is a constant tuple, $s$ is an injective tuple and no entry in $s$ equals the main value of $u$.

\begin{observation}
\label{obs:complementaryEN}
Let $R \in \{  R'_{\EN},  R'_{\NE} \}$ and $m \in [k]$.  If for all $\{E, N \}$-connected constant-injective pairs $u,s \in A^k$ with $n_{E(u,s)} = m$ 
we have $(f(u), f(s))$ in $E$ (resp.  in $N$), then either
\begin{itemize}
\item for all constant-injective pairs $u,s \in A^k$ with $n_{E(u,s)} \in \{ k - m, k-m+1 \}$
it holds $((f(u), f(s))$ in $N$ (resp. in $E$), or
\item $R$  contains either an $\EE$-tuple or a $\NN$-tuple.
\end{itemize}
\end{observation}

\begin{proof}
We only handle the case where $(f(u), f(s)) \in E$ for 
all  $\{ E, N \}$-connected constant-injective $u,s \in A^k$ with $n_{E(u,s)} = m$ and we show that in this case if $R$ does not satisfy the first item, then $R$ contains an $\EE$-tuple. 
The proof of the other case is analogous. The difference is that we show the existence of a $\NN$-tuple.
Assume that there is a $\{ E, N \}$-connected constant-injective
$u,s \in A^k$ 
such that
$n_{E(u,s)} \in  \{ (k - m), (k-m+1) \}$ and $(f(u), f(s)) \notin N$.
Since $f$ preserves a first-order expansion $\mathbb{A}$ of the random graph, it also 
preserves $E, N$ and $\neq$. Thus, we may assume that $(f(u), f(s)) \in E$.

Let now $t \in A^k$ be a tuple with the same main value as $u$ such that none of the values in $t$ occurrs in $s$,
$n_{E(t,s)} = k-m$ and $t,s$ are $\{ E, N \}$-connected. If already $n_{E(u,s)} = k-m$, then we take $t$ to be $u$, otherwise $n_{E(u,s)} = k-m +1$ and then using the extension property we simply change a value of one of the entries in $u$ so that the obtained $t$ satisfy the desired properties.
Further, either Case~\ref{ENtuples:eqconnected} or Case~\ref{ENtuples:injective} in Observation~\ref{obs:ENtuples} holds.
 
For Case~\ref{ENtuples:eqconnected} select a constant tuple $u'$ with the main value $d$ such that 
for all $i \in [k]$ the tuple $(t[i], s[i], s[i], u'[i])$ equals either $t_{\EN}$ or $t_{\NE}$. By 
Observation~\ref{obs:oneequalityEN} we may choose $u'$ such that $d$ is different from all values in $t$ and $s$
and hence the existence of $u'$ follows by the extension property of the random graph. In particular, we have that 
$n_{E(s, u')} = k$. By the assumption, it follows that $(f(s),f(u')) \in E$. On the other hand, since 
$t$ has the same main value as $u$, it follows that $(f(t), f(s)) \in E$. The tuples 
$t_{\EN}$ and $t_{\NE}$ are in $R$ and hence $(t[i], s[i], s[i], u'[i])$ is in $R$ for every $i \in [k]$ which is pp-definable in $\mathbb{A}$ and hence preserved by $f$. It follows that $R$ contains an $\EE$-tuple. It completes the proof of the claim for Case~\ref{ENtuples:eqconnected} in Observation~\ref{obs:ENtuples}. In the proof for 
Case~\ref{ENtuples:injective}  we choose
first an injective $s'$ with values not occurring in $t,s$ and then a constant $u'$ with the main value $d$ which again does not occur in none of $t,s,s'$ so that 
for all $i \in [k]$ the tuple $(t[i], s[i], s'[i], u'[i])$ equals either $t_{\EN}$ or $t_{\NE}$.
Since in this case both $t_{\EN}$ and $t_{\NE}$ are injective we can clearly find first such  $s'$and $u$ 
using the extension property of the random graph.
 Then using the same reasoning as in Case~\ref{ENtuples:eqconnected} we show that also in 
Case~\ref{ENtuples:injective} we obtain an $\EE$-tuple in $R$. It completes the proof of the observation.
\end{proof}

We continue the proof of the lemma. Since  $f$ preserves $E$, we have that 
$(f(u),f(s)) \in E$ for all constant injective pairs $u,s \in A^k$ with  
$n_{E(u,s)} = k$. By one application of 
Observation~\ref{obs:complementaryEN}, it follows that
either $R$ contains 
the desired $\EE$-tuple or 
$(f(u),f(s)) \in N$ for all constant-injective pairs $u,s$ with $n_{E(u,s)} =1$.
In the latter case, 
by another application of the claim we have that $R$ contains the desired $\NN$-tuple or
$(f(u),f(s)) \in E$ for all 
constant-injective pairs $u,s$ with $n_{E(u,s)} = n-1$.
By inductively continuing this process we obtain that 
$(f(u),f(s)) \in E$ for all constant-injective pairs 
$u, s$ 
with $n_{E(u,s)} = i$ where $i \in [k]$,
and on the other hand that
$(f(u),f(s)) \in N$ for all constant-injective pairs 
$u, s$ 
with $n_{E(u,s)} = k - i$ where $i \in [k]$ or 
$R$ contains either an $\EE$-tuple or a $\NN$-tuple.
Since  $(f(u),f(s))$ cannot be in both  $E$ and $N$,
we obtain that either an $\EE$-tuple or a $\NN$-tuple is in $R$. 
It follows that both $R'_{\EN}$ and $R'_{\NE}$ contain either an $\EE$-tuple
or a $\NN$-tuple. Hence either $R'_{\EN}$ does not efficiently entail $(E(x_1, x_2) \rightarrow \uuN(x_3, x_4))$
or $R'_{\EN}$ does not efficiently entail $(E(x_1, x_2) \rightarrow \uuN(x_3, x_4))$. It implies the 
contradiction and  completes the proof of the lemma.
$\square$

\section{Proof of Lemma~\ref{lem:noEEqEqEImplications}}

Assume on the contrary that $\mathbb{A}$ pp-defines $R$ which is either a
$[(E(x_1, x_2) \implies (x_3 = x_4),(\uuE(x_1, x_2) \wedge N(x_2, x_3) \wedge \uuE(x_3, x_4))]\textrm{-relation}$ 
or a $[((x_1 = x_2) \implies E(x_3, x_4)],(\uuE(x_1, x_2) \wedge N(x_2, x_3) \wedge \uuE(x_3, x_4))]\textrm{-relation}$.
In any case $R$ contains both an $\EEq$-tuple $t_{E=}$ and an $\EqE$-tuple $t_{=E}$
such that for all $t \in \{  t_{E=}, t_{=E} \}$ we have $(t[2], t[3]) \in N$. 
In the remainder of the proof we show that $R$ preserved by $h$ and the qnu-operation $f$ from the formulation of the lemma containing both these tuples contains also an $\EE$-tuples $t_{\EE}$
and a $==$-tuple $t_{==}$. Observe that to this end it is enough to show that at least one of these tuples is in $R$.
Indeed, $h(t_{E=}, t_{=E}, t_{EE})$ produces an $==$-tuple and  
$h(t_{E=}, t_{=E}, t_{==})$ produces an $EE$-tuple. A relation that contains 
$t_{\EE}, t_{E=}, t_{=E}$ and $t_{==}$
efficiently entails netiher  $(E(x_1, x_2) \implies (x_3 = x_4))$ nor
$((x_1 = x_2) \implies E(x_3, x_4))$. It contradicts the assumption that $R$ does efficiently entail one of these relations and will complete the proof of the lemma.
We let $k$ to be the arity of the qnu-operation $f$  and continue with a simple observation. 

\begin{observation}
\label{obs:compEEqEqE}
Let $m \in [k]$.  If for all $\{ E, = \}$-connected $u, s \in A^k$ such that $u$ is constant and $n_{E(u,s)} = m$ 
it holds that $(f(u), f(s)) \in E$ (resp. $(f(u) = f(s))$). Then either
\begin{itemize}
\item for   
all $\{ E, = \}$-connected $u, s \in A^k$ such that $u$ is constant and $n_{E(u,s)} = \{ k - m, k-m+1 \}$
we have $f(u) = f(s)$ (resp. $(f(u), f(s)) \in E$), or
\item $R$ contains an $\EE$-tuple (resp. an $==$-tuple).
\end{itemize}
\end{observation}

\begin{proof}
We only handle the case where $(f(u), f(s)) \in E$ for 
all  $\{ E, = \}$-connected $u,s \in A^k$ such that $u$ is constant and 
$n_{E(u,s)} = m$. Assume that there exists a constant $u \in A^k$ and $s \in A^k$
such that $u,s$ are $\{ E, =  \}$-connected, $n_{E(u,s)} \in \{ k - m, k-m+1 \}$,
and $(f(s) \neq f(u))$. 
Since $f$ preserves $R$ and therefore $\uuE$, we may assume that $(f(u), f(s)) \in E$.
We will show that in this case $R$ contains an $\EE$-tuple.
Select $t$ to be a $k$-tuple with the same main value as $u$ and such that 
$n_{E(t,s)} = k - m$. Observe that all entries in $t, s$ comes from a single edge in $C^2_{\omega}$.
Then, clearly, there exists $s', u' \in A^k$ of values originating from any other edge in $C^2_{\omega}$
such that $n_{E(u',s')} =m$ and such that 
for all $i \in [k]$ it holds that 
$(t[i], s[i], s'[i], u'[i])$ is either $t_{E=}$ or $t_{=E}$.
By the assumption, $(f(s'), f(u')) \in E$, and hence $(f(t), f(s),f(s'), f(u'))$ is an $\EE$-tuple in $R$. It completes the proof of the claim.
\end{proof}

Since $f$ preserves $E$, we have that 
$(f(u),f(s)) \in E$ for all $u,s \in A^k$ with  
$n_{E(u,s)} = k$. By one application of 
Observation~\ref{obs:compEEqEqE}, it follows that
either $R$ contains 
the desired $\EE$-tuple or 
$(f(u) = f(s))$ for all 
$\{ E, = \}$-connected $u,s \in A^k$ where $u$
is constant and $n_{E(u,s)} =1$.
If the application did not return the desired tuple, then
by another application of the claim we have that 
$(f(u),f(s)) \in E$ for all $\{ E, = \}$-connected 
$u, s \in A^k$ such that $u$ is constant
and $n_{E(u,s)} = n - 1$. 
By inductively continuing this process we obtain that 
$(f(u),f(s)) \in E$ for all $\{ E, = \}$-connected $u,s \in A^k$ such that $u$ is constant 
and  $n_{E(u,s)} = i$ where $i \in [k]$,
and on the other hand that
$(f(u) = f(s))$ for all constant 
$u$ and all $s$ 
with $n_{E(u,s)} = k - i$ where $i \in [k]$.
Since  $(f(u),f(s))$ cannot be in both  $E$ and $=$,
we obtain a contradiction. 
Thus, in consequence $R$ contains all: an $\EE$-tuple and an $==$-tuple. It implies a contradiction and completes   
the proof of the lemma.
$\square$

\section{Proof of Lemma~\ref{lem:noEEEqEqImplications}}

The proof goes along the lines of the proof of Lemma~\ref{lem:noEEqEqEImplications}.
We assume on the contrary that $\mathbb{A}$ pp-defines $R$ which is either a
$[(E(x_1, x_2) \implies (E(x_3, x_4)),(\uuE(x_1, x_2) \wedge N(x_2, x_3) \wedge \uuE(x_3, x_4))]$-relation or
 a  $[((x_1 = x_2) \implies (x_3 = x_4),(\uuE(x_1, x_2) \wedge N(x_2, x_3) \wedge \uuE(x_3, x_4))]$-relation 
and observe that in any case $R$ 
contains both an $\EE$-tuple $t_{\EE}$ and a $==$-tuple $t_{==}$ such that for all $t \in \{ t_{\EE}, t_{==} \}$
we have $(t[2], t[3]) \in N$. We will show that such $R$ preserved by 
the oligopotent qnu-operation $f$ from the formulation of the lemma and $h$
contains additionally an $\EEq$-tuple $t_{E=}$ and an $\EqE$-tuple $t_{=E}$. In fact, it is enough to show that $R$
contains one of these tuples. Indeed, $h(t_{EE}, t_{==}, t_{E=})$ produces an $\EqE$-tuple and 
$h(t_{EE}, t_{==}, t_{=E})$ an $\EEq$-tuple. A relation containing 
$\{ t_{EE}, t_{==}, t_{E=}, t_{=E} \}$ efficiently entails neither $(E(x_1, x_2) \implies (E(x_3, x_4))$
nor  $((x_1 = x_2) \implies (x_3 = x_4))$. It contradicts the assumption and will complete the proof of the lemma.
We let $k$ to be the arity of $f$ and continue with a simple observation.

\begin{observation}
\label{obs:compEEEqEq}
Let $m \in [k]$. If for all $\{ E, = \}$-connected $u, s \in A^k$ such that $u$ is constant and $n_{E(u,s)} = m$ it holds that 
$(f(u), f(s)) \in E$. Then either
\begin{itemize}
\item for   
all $\{ E, = \}$-connected $u, s \in A^k$ such that $u$ is constant and $n_{E(u,s)} = m-1$
we have $(f(u), f(s)) \in E$, or
\item $R$ contains an $\EEq$-tuple.
\end{itemize}
\end{observation}

\begin{proof}
Assume on the contrary that there exists a constant tuple $u  \in A^k$ and $s \in A^k$
such that $u,s$ are $\{ E, = \}$-connected, $n_{E(u,s)} = m -1$
and $(f(s), f(u)) \notin E$. Since $f$ preserves $R$ and hence $\uuE$, we have 
$(f(s) = f(u))$.
We will show that in this case $R$ contains an $\EEq$-tuple.
Select $t$ to be a $k$-tuple with the same main value as $u$ and such that 
$n_{E(t,s)} =m$. Observe that all entries in $t, s$ come from a single edge in $C^2_{\omega}$.
Then, clearly, there exists $s', u' \in A^k$ of values originating from any other edge in $C^2_{\omega}$
such that $n_{E(s',u')} =  m$ and such that 
for all $i \in [k]$ it holds that 
$(u'[i], s'[i], s[i], t[i])$ is either $t_{\EE}$ or $t_{==}$. 
Since $t$ has the same main-value as $u$ we have $(f(s) = f(t))$. Since $n_{E(s',u')} =  m$,
by the assumption, $(f(s'), f(u')) \in E$. It follows that 
$(f(u'), f(s'),f(s), f(t))$ is an $\EEq$-tuple in $R$. It completes the proof of the observation.
\end{proof}

The operation $f$ preserves $E$ and hence $(f(u),f(s)) \in E$ for all $u,s \in A^k$ with $n_{E(u,s)} = k$.
By Observation~\ref{obs:compEEEqEq}, it follows that
either $R$ contains an $\EEq$-tuple or
$(f(u), f(s)) \in E$ for all constant $u \in A^k$ and $s \in A^k$ such that $u,s$ are $\{ E, = \}$-connected and  $n_{E(u,s)} = k-1$.
By inductively continuing this argument, we either obtain  the desired  
$\EEq$-tuple or we 
have $(f(u), f(s)) \in E$ for all constant $u \in A^k$ and $s \in A^k$ such that $u,s$ are $\{ E, = \}$-connected and  $n_{E(u,s)} = k-1, k-2, \ldots, 1$.
The latter contradicts the fact that $f$ is a quasi near-unanimity operation. Indeed, for any constant $u$
and $t$ with the same main-value as $u$ which are  $\{ E, = \}$-connected and $n_{E(u,s)} = 1$, we have
$(f(u), f(s)) \in E$ and in consequence $(f(u) \neq f(s))$. 
 It completes the proof of the lemma.
$\square$

\end{document}